\documentclass[journal]{IEEEtran}
\ifCLASSINFOpdf
\else
\fi
\usepackage{amsmath,amsfonts,amssymb,color}
\usepackage[dvipsnames]{xcolor}
\usepackage{dsfont}
\definecolor{winered}{rgb}{0.8,0,0}
\definecolor{myblue}{rgb}{0,0,0.8}
\definecolor{mygreen}{rgb}{0.0 0.5 0}
\usepackage{tikz}
\usepackage{amsthm}
\usepackage{empheq}

\newtheorem{definition}{Definition}
\newtheorem{theorem}{Theorem}
\newtheorem{lemma}{Lemma}
\newtheorem{proposition}{Proposition}
\newtheorem{corollary}{Corollary}
\newtheorem{remark}{Remark}
\newtheorem{assumption}{Assumption}
\newtheorem{example}{Example}
\DeclarePairedDelimiter\ceil{\lceil}{\rceil}

\DeclareMathOperator*{\argmax}{\arg\!\max}
\DeclareMathOperator*{\argmin}{\arg\!\min}

\usepackage{algorithm}
\usepackage{algpseudocode}

\usepackage{epsfig}
\usepackage{array}
\usepackage{multirow}
\usepackage{epstopdf}
\usepackage{tikz}
\usepackage{relsize}
\usetikzlibrary{shapes,arrows}
\usepackage[hidelinks]{hyperref}
\hypersetup{
    colorlinks=true,
    linkcolor={winered},
    citecolor={mygreen}
}
\usepackage{cite}

\begin{document}
\title{Distributed Inference with Sparse and Quantized Communication}
\author{Aritra Mitra, John A. Richards, Saurabh Bagchi, and Shreyas Sundaram
\thanks{A. Mitra, S. Bagchi and S. Sundaram are with the School of Electrical and Computer Engineering at Purdue University. J. A.  Richards is with Sandia National Laboratories.   Email: {\tt \{mitra14, sbagchi, sundara2\}@purdue.edu},  {\tt{jaricha@sandia.gov}}. This work was supported in part by NSF CAREER award
1653648, and by the Laboratory Directed Research and Development program at Sandia National Laboratories. Sandia National Laboratories is a multimission laboratory managed and operated by National Technology \& Engineering Solutions of Sandia, LLC, a wholly owned subsidiary of Honeywell International Inc., for the U.S. Department of Energy's National Nuclear Security Administration under contract DE-NA0003525. The views expressed in the article do not necessarily represent the views of the U.S. Department of Energy or the United States Government.}}
\maketitle
\begin{abstract}
We consider the problem of distributed inference where agents in a network observe a stream of private signals generated by an unknown state, and aim to uniquely identify this state from a finite set of hypotheses. We focus on scenarios where communication between agents is costly, and takes place over channels with finite bandwidth. To reduce the frequency of communication, we develop a novel event-triggered distributed learning rule that is based on the principle of diffusing low beliefs on each false hypothesis. Building on this principle, we design a trigger condition under which an agent broadcasts only those components of its belief vector that have adequate innovation, to only those neighbors that require such information. We prove that our rule guarantees convergence to the true state exponentially fast almost surely despite sparse communication, and that it has the potential to significantly reduce information flow from uninformative agents to informative agents. Next, to deal with finite-precision communication channels, we propose a distributed learning rule that leverages the idea of adaptive quantization. We show that by sequentially refining the range of the quantizers, every agent can learn the truth exponentially fast almost surely, while using just $1$ bit to encode its belief on each hypothesis. For both our proposed algorithms, we rigorously characterize the trade-offs between communication-efficiency and the learning rate. 
\end{abstract}
\section{Introduction}
Over the last couple of decades, there has been a significant shift in the model of computation - driven in part by the nature of emerging applications, and partly due to concerns of reliability and scalability - from that of a single centralized computing node to parallel, distributed architectures comprising of several devices. Depending upon the context, these devices could be smart phones interacting with the cloud in a Federated Learning setup, or embedded sensors in a modern Internet of Things (IoT) network. Typically, the devices in the above applications - henceforth referred to as \textit{agents} - run on limited battery power, and setting up communication links between such agents incurs significant latency. Thus, the need arises to reduce the amount of communication to achieve a given objective. Moreover, the communication links themselves have finite bandwidth, dictating the need to compress messages appropriately. These \textit{communication bottleneck}s pose a major technical challenge. Our goal in this paper is to take a step towards resolving this challenge for the canonical problem of distributed inference. We now briefly describe this problem. 

Consider a network of agents, where each agent receives a stream of private signals sequentially over time. The observations of each agent are generated by a common underlying distribution, parameterized by an unknown static quantity which we call the \textit{true state of the world}. The task of the agents is to collectively identify this unknown quantity from a finite family of hypotheses, while relying solely on local interactions. The problem described above arises in a variety of scenarios ranging from detection and object recognition using autonomous robots, to statistical inference and learning over multiple processors. As such, the distributed inference/hypothesis testing problem enjoys a rich history  \cite{jad1,jad2,liu,salami,shahin,nedic,lalitha,su1,uribe,mitraACC19,mitra2019new}, where a variety of techniques have been proposed over the years, with more recent efforts directed towards improving the convergence rate. These techniques can be broadly classified in terms of the data aggregation mechanism: while consensus-based linear \cite{jad1,jad2,liu,salami} and log-linear \cite{shahin,nedic,lalitha,su1,uribe} rules have been well studied, \cite{mitraACC19} and \cite{mitra2019new} propose a min-protocol that leads to improved asymptotic learning rates over previous approaches.

In general, for the problem described above, no one agent can eliminate every false hypothesis on its own to uniquely learn the true state. This leads to a fundamental tension: although communication is costly (due to battery power constraints) and imprecise (due to finite communication bandwidth), it is also necessary. \textit{How should the agents interact to learn the true state despite sparse and imprecise communication?} At the moment, a theoretical understanding of this question is lacking in the distributed inference literature. In this context, our main contributions are described below. 

\subsection{Contributions}
To reduce the frequency of communication, one needs to first answer a few basic questions. (i) When should an agent exchange information with a neighbor? (ii) What piece of information should the agent exchange? To address these questions in a principled way, our first contribution is to develop a novel distributed learning rule in Section \ref{sec:algo} by drawing on ideas from  event-triggered control \cite{tabuadaevent,heemels}. The premise of our rule is based on diffusing low beliefs on each false hypothesis across the network. Building on this principle, we design a trigger condition that carefully takes into account the specific structure of the problem, and enables an agent to decide, using purely local information, whether or not to broadcast its belief\footnote{By an agent's ``belief vector", we mean a distribution over the set of hypotheses; this vector gets recursively updated over time as an agent acquires more information.} on a given hypothesis to a given neighbor. Specifically, based on our event-triggered strategy, an agent broadcasts \textit{only} those components of its belief vector that have adequate ``innovation", to \textit{only} those neighbors that are in need of the corresponding pieces of information. Thus, our approach not only reduces the  communication frequency, but also the amount of information transmitted in each round. 

Our second contribution is to provide a detailed theoretical characterization of the proposed event-triggered learning rule in Section \ref{sec:results}. Specifically, in Theorem \ref{thm:main} we establish that our rule enables each agent to learn the true state exponentially fast almost surely, under standard assumptions on the observation model and the network topology. We characterize the learning rate of our algorithm as a function of the agents' relative entropies, the network structure, and parameters of the communication model. In particular, we show that even when the inter-communication intervals between the agents grow geometrically at a rate $p > 1$, our rule guarantees exponentially fast learning at a network-dependent rate that scales inversely with $p$. However, when such intervals grow polynomially, the learning rate remains the same as the network-independent learning rate in \cite{mitra2019new}. Thus, our results provide various interesting insights into the relationship that exists between the rate of convergence and the sparsity of the communication pattern. 

Next, in Proposition \ref{prop:sparsity} and Corollary  \ref{prop:tree}, we demonstrate that our event-triggered scheme has the potential to significantly reduce information flow from uninformative agents to informative agents.  Finally, in Theorem \ref{thm:asymp}, we argue that if asymptotic learning of the true state is the only consideration, then one can allow for communication schemes with arbitrarily long intervals between successive communications. 

While our results above concern the aspect of sparse communication, in Section \ref{sec:quantrule} we turn our attention to learning over communication channels with finite precision, i.e., channels that can support only a finite number of bits. In a recent paper \cite{lalitha} that looks at the same problem as us, the authors demonstrated in simulations that with a quantized variant of their log-linear rule, the beliefs of the agents might converge to a wrong hypothesis, if not enough bits are used to encode the beliefs. It is natural to then ask whether the above phenomenon is to be expected of $any$ rule, or whether it is specific to the one explored in \cite{lalitha}. We argue that it is in fact the latter by resolving the following fundamental question. 
\textit{In order to learn the true state, how many bits must an agent use to encode its belief on each hypothesis?} To answer this question, we develop a distributed learning rule based on the idea of adaptive quantization. The key feature of our rule is to successively refine the range of the quantizers as the agents acquire more information over time and narrow down on the truth. In Theorem \ref{thm:Quantmin}, we prove that even if every agent uses just 1 bit to encode its belief on each hypothesis, all agents end up learning the truth exponentially fast almost surely. The rate of learning, however, exhibits a dependence on the precision of the quantizer - a dependence that we explicitly characterize. In doing so, we show that if the number of bits used for encoding each hypothesis is chosen to be large enough w.r.t. certain relative entropies, then one can recover the exact same long-run learning rate as with infinite precision, i.e., the rate obtained in \cite{mitra2019new}. This constitutes our final contribution. 

To summarize, this paper (i) develops novel communication-efficient distributed inference  algorithms; (ii) provides detailed theoretical characterizations of their performance; and, in particular, (iii) highlights various interesting trade-offs between sparse and imprecise communication, and the learning rate.  

This paper significantly expands upon our preliminary work in  \cite{mitraCDC20} where we only consider the effect of sparse communication. In particular, Sections \ref{sec:quantrule} -- \ref{sec:quant_event} that deal with the aspect of imprecise communication are  entirely new additions. 

\subsection{Related Work}
Our work is closely related to the papers \cite{switching} and \cite{hareiccasp}, each of which explores the theme of event-driven communications for  distributed learning. In \cite{switching}, the authors propose a rule where an agent queries the log-marginals of its neighbors only if the total variation distance between its current belief and the Bayesian posterior after observing a new signal falls below a pre-defined threshold. That is, an agent communicates only if its current private signal is not adequately informative. Among various other differences, the trigger condition we propose is not only a function of an agent's local observations, but also carefully incorporates feedback from neighboring agents. Moreover, while we provide theoretical results to substantiate that our rule leads to sparse communication patterns, \cite{switching} does so only via simulations. The algorithm in \cite{hareiccasp} comes with no theoretical guarantees of convergence. 

The aspect of sparse communication has been studied in the context of a variety of coordination problems on networks, such as average consensus \cite{olshevsky}, distributed optimization \cite{opt1,opt2}, and static parameter estimation \cite{sahu} - settings that differ from the one we investigate in this paper. To promote communication-efficiency, \cite{opt1} and \cite{sahu} propose algorithms where inter-agent interactions become progressively sparser over time. However, these algorithms are essentially time-triggered, i.e., they do not adhere to the principle that ``an agent should communicate only when it has something useful to say". On the other hand, the strand of literature that deals with event-driven communications for multi-agent systems focuses primarily on variations of the basic consensus problem; we refer the reader to \cite{nowzari} for a survey of such techniques. 

Our work is also related to the classical literature on decentralized hypothesis testing under communication constraints \cite{csiszar,longo,amari}. However, unlike our formulation, these papers assume the presence of a centralized fusion center, and do not deal with sequential data, i.e., each agent only receives one signal. Finally, the adaptive quantization idea  used in this paper bears conceptual similarities to the encoding strategy in \cite{tatikonda} for stabilizing an LTI plant over a bit-constrained channel, and also to a recent work on distributed optimization \cite{thinh}. 

\section{Model and Problem Formulation}
\label{sec:model}
\textbf{Network Model:} We consider a group of agents $\mathcal{V}=\{1,\ldots,n\}$, and model interactions among them via an undirected graph $\mathcal{G}=(\mathcal{V},\mathcal{E})$. An edge $(i,j)\in\mathcal{E}$ indicates that agent $i$ can directly transmit information to agent $j$, and vice versa. The set of all neighbors of agent $i$ is defined as $\mathcal{N}_i=\{j\in\mathcal{V}:(j,i)\in\mathcal{E}\}$. We say that $\mathcal{G}$ is rooted at $\mathcal{C}\subseteq\mathcal{V}$, if for each agent $i\in\mathcal{V}\setminus\mathcal{C}$, there exists a path to it from some agent $j\in\mathcal{C}$. For a connected graph $\mathcal{G}$, we will use $d(i,j)$ to denote the length of the shortest path between $i$ and $j$. 

\textbf{Observation Model:} Let $\Theta=\{\theta_1,\theta_2,\ldots,\theta_m\}$ denote $m$ possible states of the world, with each state representing a hypothesis. A specific state $\theta^{\star}\in\Theta$, referred to as the true state of the world, gets realized. Conditional on its realization, at each time-step $t\in\mathbb{N}_{+}$, every agent $i\in\mathcal{V}$  privately observes a signal $s_{i,t}\in\mathcal{S}_i$, where $\mathcal{S}_i$ denotes the signal space of agent $i$.\footnote{We use $\mathbb{N}$ and $\mathbb{N}_{+}$ to represent the set of non-negative integers and positive integers, respectively.} The joint observation profile generated across the network is denoted ${s}_{t}=(s_{1,t},s_{2,t},\ldots,s_{n,t})$, where $s_t\in\mathcal{S}$, and $\mathcal{S}=\mathcal{S}_1\times\mathcal{S}_2\times\ldots \mathcal{S}_n$. 
Specifically, the signal $s_{t}$ is generated based on a conditional likelihood function $l(\cdot|\theta^{\star})$, the $i$-th marginal of which is denoted $l_i(\cdot|\theta^{\star})$, and is available to agent $i$. The signal structure of each agent $i\in\mathcal{V}$ is thus characterized by a family of parameterized marginals $l_i=\{l_i(w_i|\theta): \theta\in\Theta, w_i\in\mathcal{S}_i\}$. We make certain standard assumptions \cite{jad1,jad2,liu,salami,shahin,lalitha,su1,mitraACC19,mitra2019new}: (i) The signal space of each agent $i$, namely $\mathcal{S}_i$, is finite. (ii) Each agent $i$ has knowledge of its local likelihood functions $\{l_i(\cdot|\theta_p)\}_{p=1}^{m}$, and it holds that $l_i(w_i|\theta) > 0, \forall w_i\in\mathcal{S}_i$, and $\forall \theta \in \Theta$. (iii) {The observation sequence of each agent is described by an i.i.d. random process over time; at each time-step, agents make independent observations.} (iv) There exists a fixed true state of the world $\theta^{\star}\in\Theta$ (unknown to the agents) that generates the observations of all the agents. The probability space for our model is denoted $(\Omega,\mathcal{F},\mathbb{P}^{\theta^{\star}})$, where $\Omega\triangleq\{\omega: \omega=(s_1,s_2,\ldots), \forall s_t\in\mathcal{S}, \forall t \in \mathbb{N}_{+}\}$, $\mathcal{F}$ is the $\sigma$-algebra generated by the observation profiles, and $\mathbb{P}^{\theta^{\star}}$ is the probability measure induced by sample paths in $\Omega$. Specifically, $\mathbb{P}^{\theta^{\star}}=\prod \limits_{t=1}^{\infty}l(\cdot|\theta^{\star})$. We will use the abbreviation a.s. to indicate almost sure occurrence of an event w.r.t. $\mathbb{P}^{\theta^{\star}}$.

The goal of each agent in the network is to eventually learn the true state  $\theta^{\star}$. However, the key challenge in achieving this objective arises from an \textit{identifiability problem} that each agent might potentially face. To make this precise, define $\Theta^{\theta^{\star}}_i\triangleq\{\theta\in\Theta : l_i(w_i|\theta)=l_i(w_i|\theta^{\star}), \forall w_i\in\mathcal{S}_i\}$. In words, $\Theta^{\theta^{\star}}_i$ represents the set of hypotheses that are \textit{observationally equivalent} to $\theta^{\star}$ from the perspective of agent $i$. Thus, if $|\Theta^{\theta^{\star}}_i| > 1$, it will be impossible for agent $i$ to uniquely learn the true state $\theta^*$ without interacting with its neighbors. 

Our broad \textbf{goal} in this paper is to develop distributed learning algorithms that resolve the identifiability problem described above despite  \textit{sparse} and \textit{imprecise} communication. {To this end, we will first separately explore the ideas of event-triggering for sparse communication, and adaptive quantization for imprecise communication, in Sections \ref{sec:algo} and \ref{sec:quantrule}, respectively. We do so to reveal in a clear, understandable way the main ideas underlying each of our algorithms. Later, in Section \ref{sec:quant_event}, we will see how these ideas can be effectively combined.} Let us begin by recalling the following definition from \cite{mitraACC19}.

\begin{definition}(\textbf{Source agents}) An agent $i$ is said to be a source agent for a pair of distinct hypotheses $\theta_p,\theta_q\in\Theta$ if it can distinguish between them, i.e., if  $D(l_i(\cdot|\theta_p)||l_i(\cdot|\theta_q)) > 0$, where $D(l_i(\cdot|\theta_p)||l_i(\cdot|\theta_q))$ represents the KL-divergence \cite{cover} between the distributions $l_i(\cdot|\theta_p)$ and $l_i(\cdot|\theta_q)$. The set of source agents for pair $(\theta_p,\theta_q)$ is denoted $\mathcal{S}(\theta_p,\theta_q)$.
\end{definition}

Throughout the rest of the paper, we will use $K_i(\theta_p,\theta_q)$
as a shorthand for $D(l_i(\cdot|\theta_p)||l_i(\cdot|\theta_q))$. {For our analysis, we will make the following standard assumption.

\begin{assumption} (\textbf{Global Identifiability}) For every pair $\theta_p,\theta_q \in \Theta$ such that $\theta_p \neq \theta_q$, the corresponding source set $\mathcal{S}(\theta_p,\theta_q)$ is non-empty. 
\label{assump:globiden}
\end{assumption}

Note that global identifiability implies  $\bigcap_{i\in\mathcal{V}} \Theta^{\theta^{\star}}_i = \{\theta^*\}$, i.e., the collective information dispersed across the network allows one to distinguish $\theta^*$ from every $\theta\neq \theta^*$.}
\color{black}{}

\section{An Event-Triggered Distributed Learning Rule}
\label{sec:algo}
$\bullet$ \textbf{Belief-Update Strategy}: In this section, we develop an event-triggered distributed learning rule that enables each agent to eventually learn the truth, despite infrequent information exchanges with its neighbors. Our approach requires each agent $i$ to maintain a local belief vector $\boldsymbol{\pi}_{i,t}$, and an actual belief vector $\boldsymbol{\mu}_{i,t}$, each of which are probability distributions over the hypothesis set $\Theta$, and hence of dimension $m$. While agent $i$ updates $\boldsymbol{\pi}_{i,t}$ in a Bayesian manner using only its private signals (see Eq. \eqref{eqn:Bayes}), to formally describe how it updates  $\boldsymbol{\mu}_{i,t}$, we need to first introduce some notation. Accordingly, let $\mathds{1} _{ji,t}(\theta)\in\{0,1\}$ be an indicator variable which takes on a value of 1 if and only if agent $j$  broadcasts $\mu_{j,t}(\theta)$ to agent $i$ at time $t$. Next, we define $\mathcal{N}_{i,t}(\theta)\triangleq\{j\in\mathcal{N}_i|\mathds{1}_{ji,t}(\theta)=1\}$ as the subset of agent $i$'s neighbors who broadcast their belief on $\theta$ to $i$ at time $t$. As part of our learning algorithm, each agent $i$ keeps track of the lowest belief on each hypothesis $\theta\in\Theta$ that it has heard up to any given instant $t$, denoted by $\bar{\mu}_{i,t}(\theta)$. More precisely,  $\bar{\mu}_{i,0}(\theta)=\mu_{i,0}(\theta)$, and $ \forall t\in\mathbb{N}$, 
\begin{equation}
    \bar{\mu}_{i,t+1}(\theta)=\min\{\bar{\mu}_{i,t}(\theta),\{\mu_{j,t+1}(\theta)\}_{j\in \{i\}\cup \mathcal{N}_{i,t+1}(\theta)}\}.
    \label{eq:mubardefn}
\end{equation}
We are now in position to describe the belief-update rule at each agent: $\boldsymbol{\pi}_{i,t}$ and $\boldsymbol{\mu}_{i,t}$ are initialized with $\pi_{i,0}(\theta)>0,\mu_{i,0}(\theta)>0, \forall \theta\in\Theta,\forall i\in\mathcal{V}$ (but otherwise arbitrarily), and subsequently updated as follows $\forall t  \in\mathbb{N}$: 
\begin{equation}
\pi_{i,t+1}(\theta)=\frac{l_i(s_{i,t+1}|\theta)\pi_{i,t}(\theta)}{\sum  \limits_{p=1}^{m} l_i(s_{i,t+1}|\theta_p)\pi_{i,t}(\theta_p)},
\label{eqn:Bayes}
\end{equation}
\begin{equation}
\mu_{i,t+1}(\theta)=\frac{\min\{\bar{\mu}_{i,t}(\theta),\pi_{i,t+1}(\theta)\}}{\sum\limits_{p=1}^{m}\min\{\bar{\mu}_{i,t}(\theta_p),\pi_{i,t+1}(\theta_p)\}}.
\label{eqn:update}
\end{equation}

\begin{algorithm}[t]
\caption{\textbf{(Event-Triggered Min-Rule)}  Each agent $i \in \mathcal{V}$ executes this algorithm in parallel} \label{algo:ETmin}
\textbf{Initialization:}  $\mu_{i,0}(\theta)>0$,  $\pi_{i,0}(\theta)>0$, $\bar{\mu}_{i,0}(\theta)=\mu_{i,0}(\theta),\forall \theta\in\Theta$, and  $\sum_{\theta\in\Theta}\mu_{i,0}(\theta)=1$, $\sum_{\theta\in\Theta}\pi_{i,0}(\theta)=1$.
\begin{algorithmic}[1]
\For {$t \in \mathbb{N}$} 
\For {$\theta\in\Theta$}

\State Update $\pi_{i,t+1}(\theta)$ via \eqref{eqn:Bayes}, and   $\mu_{i,t+1}(\theta)$ via  \eqref{eqn:update}. 

\If {$t+1=t_1$}

\State Broadcast $\mu_{i,t+1}(\theta)$ to each $j\in\mathcal{N}_i$.

\Else

\State For each $j\in\mathcal{N}_i$, broadcast $\mu_{i,t+1}(\theta)$ to $j$ if and only if $t+1\in\mathbb{I}$ \textit{and} the event condition  \eqref{eq:event} holds. 
\EndIf
\State Receive $\mu_{j,t+1}(\theta)$ from each $j\in\mathcal{N}_{i,t+1}(\theta)$, and update $\bar{\mu}_{i,t+1}(\theta)$ via \eqref{eq:mubardefn}. 
\EndFor
\EndFor
\end{algorithmic}
\end{algorithm}

$\bullet$  \textbf{Communication Strategy}: We now focus on specifying \textit{when} an agent broadcasts its belief on a given hypothesis to a neighbor. To this end, we first define a sequence $\mathbb{I}=\{t_1,t_2,t_3,\ldots\}\in\mathbb{N}_{+}$ of \textit{event-monitoring} time-steps, where $t_1=1$, and $t_{k+1}-t_{k}=g(k), \forall k\in\mathbb{N}_{+}.$ Here, $g:[1,\infty)\rightarrow [1,\infty)$ is a continuous, non-decreasing function that takes on integer values at integers. We will henceforth refer to $g(k)$ as the \textit{event-interval} function. At any given time $t\in\mathbb{N}_{+}$, let $\hat{\mu}_{ij,t}(\theta)$ represent agent $i$'s belief on $\theta$ the last time (excluding time $t$) it transmitted its belief on $\theta$ to agent $j$. Our communication strategy is as follows. At $t_1$, each agent $i\in\mathcal{V}$ broadcasts its entire belief vector $\boldsymbol{\mu}_{i,t}$ to every neighbor. Subsequently, at each $t_k\in\mathbb{I}, k\geq2$, $i$ transmits $\mu_{i,t_k}(\theta)$ to $j\in\mathcal{N}_i$ if and only if the following event occurs:
\begin{equation}
    \mu_{i,t_k}(\theta) < \gamma(t_k)  \min\{\hat{\mu}_{ij,t_k}(\theta), \hat{\mu}_{ji,t_k}(\theta)\},
\label{eq:event}
\end{equation}
where $\gamma:\mathbb{N}\rightarrow (0,1]$ is a non-increasing function that we will henceforth call the \textit{threshold} function. If $t\notin\mathbb{I}$, then an agent $i$ does not communicate with its neighbors at time $t$, i.e., all inter-agent interactions are restricted to time-steps in $\mathbb{I}$, subject to the trigger-condition given by \eqref{eq:event}. Notice that we have not yet specified the functional forms of $g(\cdot)$ and $\gamma(\cdot)$; we will comment on these quantities later in Section \ref{sec:results}.  

$\bullet$ \textbf{Summary}: At each time-step $t+1\in\mathbb{N}_{+}$, and for each hypothesis $\theta\in\Theta$, the sequence of operations executed by an agent $i$ is summarized as follows. (i) Agent $i$ updates its local and actual beliefs on $\theta$ via  \eqref{eqn:Bayes} and \eqref{eqn:update}, respectively. (ii) For each neighbor $j\in\mathcal{N}_i$, it decides whether or not to transmit $\mu_{i,t+1}(\theta)$ to $j$, and collects $\{\mu_{j,t+1}(\theta)\}_{ j\in\mathcal{N}_{i,t+1}(\theta)}$.\footnote{If $t+1\notin\mathbb{I}$, this step gets bypassed, and $\mathcal{N}_{i,t+1}(\theta)=\emptyset,\forall\theta\in\Theta$.} (iii) It updates $\bar{\mu}_{i,t+1}(\theta)$ via \eqref{eq:mubardefn} using the (potentially) new information it acquires from its neighbors at time $t+1$. We call the above algorithm the \texttt{Event-Triggered Min-Rule} and outline its steps in  Algorithm \ref{algo:ETmin}. 

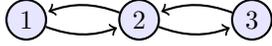
\begin{figure}[t]
\begin{center}
\begin{tikzpicture}
[->,shorten >=1pt,scale=.75,inner sep=1pt, minimum size=15pt, auto=center, node distance=3cm,
  thick, node/.style={circle, draw=black, thick},]
\tikzstyle{block1} = [rectangle, draw, fill=red!10, 
    text width=8em, text centered, rounded corners, minimum height=0.8cm, minimum width=1cm];
\node [circle, draw, fill=blue!10](n1) at (0,0)     (1)  {$1$};
\node [circle, draw, fill=blue!10](n2) at (2,0)     (2)  {$2$};
\node [circle, draw, fill=blue!10](n3) at (4,0)     (3)  {$3$};
\draw [->, thick] (1) to [bend right=20] (2);
\draw [->, thick] (2) to 
[bend right=20] (3);
\draw [->, thick] (2) to [bend right=20] (1);
\draw [->, thick] (3) to 
[bend right=20] (2);
\end{tikzpicture}
\end{center}
\caption{The figure shows a network where only agent 1 is informative. In Section \ref{sec:algo}, we design an event-triggered algorithm under which all upstream broadcasts along the path $3\rightarrow2\rightarrow1$ stop eventually almost surely. At the same time, all agents learn the true state. We demonstrate these facts both in theory (see Section \ref{sec:results}), and in simulations (see Section \ref{sec:simulations}).}
\label{fig:example}
\vspace{-3mm}
\end{figure}

$\bullet$ \textbf{Intuition:} The premise of our belief-update strategy is based on diffusing low beliefs on each false hypothesis. For a given false hypothesis $\theta$, the local Bayesian update \eqref{eqn:Bayes} will generate a decaying sequence $\pi_{i,t}(\theta)$ for each $i\in\mathcal{S}(\theta^*,\theta)$. Update rules \eqref{eq:mubardefn} and \eqref{eqn:update} then help propagate agent $i$'s low belief on $\theta$ to the rest of the network. We point out that in contrast to our earlier work \cite{mitraACC19, mitra2019new}, where for updating $\mu_{i,t+1}(\theta)$, agent $i$ used the lowest neighboring belief on $\theta$ at the \textit{previous} time-step $t$, our approach here requires an agent $i$ to use the lowest belief on $\theta$ that it has heard \textit{up to} time $t$, namely $\bar{\mu}_{i,t}(\theta)$. This modification will be crucial in the convergence analysis of Algorithm \ref{algo:ETmin}. 

To build intuition regarding our communication strategy, let us consider the network in Fig \ref{fig:example}. Suppose $\Theta=\{\theta_1,\theta_2\}, \theta^*=\theta_1$, and $\mathcal{S}(\theta_1,\theta_2)=1$, i.e., agent 1 is the only informative agent. Since our principle of learning is based on eliminating each false hypothesis, it makes sense to broadcast beliefs only if they are low enough. Based on this observation, one naive approach to enforce sparse communication could be to set a fixed low threshold, say $\beta$, and wait until beliefs fall below such a threshold to broadcast. While this might lead to sparse communication initially, in order to learn the truth, there must come a time beyond which the beliefs of all agents on the false hypothesis $\theta_2$  always stay below $\beta$, which will subsequently lead to dense communication. The obvious fix is to introduce an event-condition that is \textit{state-dependent}. Consider the following candidate strategy: an agent broadcasts its belief on a state $\theta$ only if it is sufficiently lower than what it was when it last broadcasted about $\theta$. While an improvement over the ``fixed-threshold" strategy, this new scheme has the following demerit: broadcasts are not \textit{agent-specific}. In other words, going back to our example, agent 2 (resp., agent 3) might transmit unsolicited information to agent 1 (resp., agent 2) - information, that agent 1 (resp., agent 2)  does not require. To remedy this, one can consider a request/poll based scheme as in \cite{switching} and \cite{de}, where an agent receives information from a neighbor only by polling that neighbor. However, now each time agent 2 needs information from agent 1, it needs to place a request, \textit{the request itself incurring extra communication}. 

Given the above issues, we ask: Is it possible to devise an event-triggered scheme that eventually stops unnecessary broadcasts from agent 3 to 2, and agent 2 to 1, while preserving  essential information flow from agent 1 to 2, and agent  2 to 3? \textit{More generally, we seek a triggering rule that can reduce transmissions from uninformative agents to informative agents.} This leads us to the event condition in Eq. \eqref{eq:event}. For each $\theta\in\Theta$, an agent $i$
 broadcasts $\mu_{i,t}(\theta)$ to a neighbor $j\in\mathcal{N}_i$ only if $\mu_{i,t}(\theta)$ has adequate ``innovation" w.r.t. $i$'s last broadcast about $\theta$ to $j$, \textit{and} $j$'s last broadcast about $\theta$ to $i$. A decreasing threshold function $\gamma(t)$ makes it progressively harder to satisfy the event condition in Eq. \eqref{eq:event}, demanding more innovation to merit broadcast as time progresses.\footnote{We will see later  (Corollary  \ref{prop:tree}) that for the network in Fig. \ref{fig:example}, this scheme provably stops communications from agent 3 to 2, and agent 2 to 1, eventually.} The rationale behind checking the event condition only at time-steps in $\mathbb{I}$ is twofold.\footnote{While this might appear similar to the Periodic Event-Triggering (PETM) framework \cite{PETM} where events are checked periodically, the sequence $\mathbb{I}$ can be significantly more general than a simple periodic sequence.} First, it saves computations since the event condition need not be checked all the time. Second, and more importantly, it provides an additional instrument to control communication-sparsity on top of event-triggering. Indeed, a monotonically increasing event-interval function $g(\cdot)$ implies fewer agent interactions with time, since all potential broadcasts are restricted to $\mathbb{I}$. In particular, without the event condition in Eq. \eqref{eq:event}, our communication strategy would boil down to a simple time-triggered rule, akin to the one studied in our recent work \cite{mitraCDC19}.
 
 We close this section by highlighting that our event condition (i) is \textit{$\theta$-specific}, since an agent may not be equally informative about all states\footnote{This is precisely the motivation behind tracking changes in individual components of the belief vector, as opposed to looking at changes in the overall belief vector using, for instance, the total variation metric.}; (ii) is \textit{neighbor-specific}, since not all neighbors might require information;
 (iii) is  \textit{problem-specific}, since it is built upon the principle of eliminating false hypotheses by diffusing low beliefs; and 
 (iv) can be checked using local information only. {While the event condition \eqref{eq:event} can significantly reduce communication (as we shall see in the next section), checking this condition imposes additional \textit{memory requirements} for each agent: in addition to maintaining the vectors $\boldsymbol{\pi}_{i,t}, \boldsymbol{\mu}_{i,t}$, and $\boldsymbol{\bar{\mu}}_{i,t}$, each agent $i$ has to maintain a vector $\boldsymbol{\hat{\mu}}_{ji,t}$ for each neighbor $j \in \mathcal{N}_i$. Recall that  ${\hat{\mu}}_{ji,t}(\theta)$ stores the most recent belief on $\theta$ that $i$ has received from $j$. Thus, overall, each agent $i$ needs to maintain and dynamically update $(|\mathcal{N}_i|+3)$ $m$-dimensional vectors. Note that this memory overhead need not necessarily scale with the size of the network (e.g., in sparse or bounded-degree graphs).}

\section{Theoretical Guarantees for Algorithm \ref{algo:ETmin}}
\label{sec:results}
In this section, we state the main theoretical results pertaining to our \texttt{Event-Triggered Min-Rule}, and then discuss their implications. Proofs of these results are deferred to Appendix  \ref{sec:proofs}. To state the first result concerning the convergence of our learning rule, let us define $G(z) \triangleq \int\limits_{1}^{z} g(\tau)d\tau, \forall z\in [1,\infty)$. Let $G^{-1}(\cdot)$ represent the inverse of $G(\cdot)$, i.e., $\forall z\in[1,\infty), G^{-1}(G(z))=z$. Since $g(\cdot)$ is continuous and takes values in $[1,\infty)$ by definition, $G(\cdot)$ is strictly increasing, unbounded, and continuous, with $G(1)=0$, and hence, $G^{-1}(z)$ is well-defined for all $z\in [0,\infty)$. 

\begin{theorem}
Suppose the functions $g(\cdot)$ and $\gamma(\cdot)$ satisfy:
\begin{equation} \lim_{t\to\infty}\frac{G(G^{-1}(t)-2)}{t}=\alpha\in(0,1]; \hspace{2mm} \lim_{t\to\infty}\frac{\log({1}/{\gamma(t)})}{t}=0.
\label{eqn:functions}
\end{equation}
Furthermore, suppose global identifiability (Assumption \ref{assump:globiden}) holds, and the communication graph $\mathcal{G}$ is connected. Then, Algorithm \ref{algo:ETmin} guarantees the following. 
\begin{itemize}
    \item \textbf{(Consistency)}: For each agent  $i\in\mathcal{V}$,  $\mu_{i,t}(\theta^{\star}) \rightarrow 1$ a.s.
    \item \textbf{(Exponentially Fast Rejection of False Hypotheses)}: For each agent  $i\in\mathcal{V}$, and for each false hypothesis $\theta\in\Theta\setminus\{\theta^{\star}\},$ the following holds:
    \begin{equation}
        \liminf_{t\to\infty}-\frac{\log\mu_{i,t}(\theta)}{t} \geq \max_{v\in\mathcal{S}(\theta^{\star},\theta)}\alpha^{d(v,i)}K_v(\theta^{\star},\theta) \hspace{1mm} a.s.
        \label{eqn:asymprate}
    \end{equation}
\end{itemize}
\label{thm:main}
\end{theorem}

At this point, it is natural to ask: For what classes of functions $g(\cdot)$ does the result of Theorem \ref{thm:main} hold? The following result provides an answer. 

\begin{corollary}
Suppose the conditions in Theorem \ref{thm:main} hold.
\begin{itemize}
    \item[(i)]  Suppose $g(x)=x^p, \forall x\in\mathbb{R}_{+}$, where $p$ is any positive integer. Then, for each $\theta\in\Theta\setminus\{\theta^{\star}\}$, and $i\in\mathcal{V}$:
    \begin{equation}
        \liminf_{t\to\infty}-\frac{\log\mu_{i,t}(\theta)}{t} \geq \max_{v\in\mathcal{S}(\theta^{\star},\theta)}K_v(\theta^{\star},\theta) \hspace{1mm} a.s.
        \label{eqn:corr1}
    \end{equation}
\item[(ii)] Suppose $g(x)=p^x, \forall x\in\mathbb{R}_{+}$, where $p$ is any positive integer. Then, for each $\theta\in\Theta\setminus\{\theta^{\star}\}$, and $i\in\mathcal{V}$:
\begin{equation}
        \liminf_{t\to\infty}-\frac{\log\mu_{i,t}(\theta)}{t} \geq \max_{v\in\mathcal{S}(\theta^{\star},\theta)} \frac{K_v(\theta^{\star},\theta)}{p^{2d(v,i)}} \hspace{1mm} a.s.
        \label{eqn:corr2}
    \end{equation}
\end{itemize}
\label{corr:functionalform}
\end{corollary}
\begin{proof}
The proof follows by directly computing the limit in Eq. \eqref{eqn:functions}. For case (i), $\alpha=1$, and for case (ii),  $\alpha=1/p^2$. 
\end{proof}

Clearly, the communication pattern between the agents is at least as sparse as the sequence $\mathbb{I}$. Our event-triggering scheme introduces further sparsity, as we next establish. 

\begin{proposition}
Suppose the conditions in Theorem \ref{thm:main} are met. Then, there exists $\bar{\Omega}\subseteq\Omega$ such that $\mathbb{P}^{\theta^*}(\bar{\Omega})=1$, and for each $\omega\in\bar{\Omega}$, $\exists T_1(\omega), T_2(\omega) < \infty$ such that the following hold.
\begin{itemize}
    \item[(i)] At each $t_k\in\mathbb{I}$ such that $t_k > T_1(\omega)$, $\mathds{1}_{ij,t_k}(\theta^*)\neq 1, \forall i\in\mathcal{V}$ and $\forall j\in\mathcal{N}_i$. 
    \item[(ii)] For all $\theta\neq\theta^*$, and $i \notin \mathcal{S}(\theta^*,\theta)$, it holds that at each $t_k > T_2(\omega)$, $\exists j\in \mathcal{N}_i$ such that $\mathds{1}_{ij,t_k}(\theta)\neq 1$.\footnote{In this claim, $j$ might depend on $t_k$.} 
\end{itemize}
\label{prop:sparsity}
\end{proposition}

The following result is an immediate application of the above proposition.

\begin{corollary}
Suppose the conditions in Theorem \ref{thm:main} are met. Additionally, suppose $\mathcal{G}$ is a tree graph, and for each pair $\theta_p,\theta_q \in \Theta$, $|\mathcal{S}(\theta_p,\theta_q)|=1$. Consider any $\theta\neq\theta^*$, and let $\mathcal{S}(\theta^*,\theta)=v_\theta$. Then, each agent $i\in \mathcal{V}\setminus\{v_\theta\}$ stops broadcasting its belief on $\theta$ to its parent in the tree rooted at $v_{\theta}$ eventually almost surely. 
\label{prop:tree}
\end{corollary}

A few comments are now in order.

$\bullet$ \textbf{On the nature of $g(\cdot)$ and $\gamma(\cdot)$:}  Intuitively, if the event-interval function $g(\cdot)$ does not grow too fast, and the threshold function $\gamma(\cdot)$ does not decay too fast, one should expect things to fall in place. {This intuition is made precise by the limit conditions in Eq. \eqref{eqn:functions}. In particular, the parameter $\alpha$ is a measure of how fast $g(\cdot)$ grows: roughly speaking, smaller the value of $\alpha$, the faster $g(\cdot)$ grows, as alluded to by Corollary \ref{corr:functionalform}. To achieve exponentially fast learning based on our rule, we require $\alpha$ to be strictly greater than $1$. Corollary \ref{corr:functionalform} reveals that up to integer constraints, any polynomial or exponential event-interval function meets this requirement. Regarding the threshold function, we note from \eqref{eqn:functions} that any sub-exponentially decaying $\gamma(\cdot)$ works for our purpose.}

$\bullet$ \textbf{Trade-offs between sparse communication and learning rate:} What is the price paid for sparse communication? To answer the above question, we set as benchmark the scenario studied in our previous work \cite{mitra2019new}, where we did not account for communication efficiency. There, we showed that each false hypothesis $\theta$ gets rejected exponentially fast by every agent at the \textit{network-independent} rate  $\max_{v\in\mathcal{V}} K_v(\theta^*,\theta)$.\footnote{In contrast, for linear \cite{jad1,jad2,liu,salami} and log-linear \cite{shahin,nedic,lalitha,su1,uribe} rules, the corresponding rate is a convex combination of the relative entropies $K_v(\theta^*,\theta), v\in\mathcal{V}$.} From \eqref{eqn:asymprate}, we note that under highly sparse communication regimes which correspond to $\alpha < 1$, although learning occurs exponentially fast, the learning rate gets lowered relative to \cite{mitra2019new}. Moreover, unlike \cite{mitra2019new}, \eqref{eqn:asymprate} reveals that the asymptotic learning rate is \textit{network-dependent} and \textit{agent-specific}, i.e., different agents may discover the truth at different rates. In particular, when considering the asymptotic rate of rejection of a particular false hypothesis at a given agent $i$, notice from the R.H.S. of \eqref{eqn:asymprate} that one needs to account for the attenuated relative entropies of the corresponding source agents, where the attenuation factor scales exponentially with the distances of agent $i$ from such source agents. An instance of the above scenario is when the inter-communication intervals grow geometrically at rate $p >1$; see case (ii) of Corollary \ref{corr:functionalform}.

On the other hand, from case (i) of Corollary \ref{corr:functionalform}, we glean that  polynomially growing inter-communication intervals, coupled with our proposed event-triggering strategy, lead to \textit{no loss in the long-term learning rate relative to the benchmark case in}  \cite{mitra2019new}, i.e., as far as asymptotic performance is concerned, communication-efficiency comes essentially for ``free" under this regime. {However, even when $g(x)$ grows polynomially, the \textit{transient} behavior induced by our algorithm will depend on how $g(x)$ is chosen. While in this paper we focus on the \textit{asymptotic learning rate} as our sole performance metric of interest, striking a desired balance between transient performance and sparse communication will require a finer non-asymptotic analysis of Algorithm \ref{algo:ETmin}.}

$\bullet$ \textbf{Sparse communication introduced by event-triggering:} Observe that being able to eliminate each false hypothesis is enough for learning the true state. In other words, agents need not exchange their beliefs on the true state (of course, no agent knows a priori what the true state is). Our event-triggering scheme precisely achieves this, as evidenced by claim (i) of Proposition \ref{prop:sparsity}: {on almost all sample paths, there exists a (sample-path dependent) time $T_1(\omega)$ after which every agent stops broadcasting its belief on the true state $\theta^*$.}

In addition, an important property of our event-triggering strategy is that it reduces information flow from uninformative agents to informative agents. To see this, consider any false hypothesis $\theta\neq\theta^*$, and an agent $i\notin\mathcal{S}(\theta^*,\theta)$. Since $i\notin\mathcal{S}(\theta^*,\theta)$, agent $i$'s local belief  $\pi_{i,t}(\theta)$ will stop decaying eventually, making it impossible for agent $i$ to lower its actual belief  $\mu_{i,t}(\theta)$ without the influence of its neighbors. Consequently, when left alone between consecutive event-monitoring time-steps, $i$ will not be able to leverage its own private signals to generate enough ``innovation" in $\mu_{i,t}(\theta)$ to broadcast to the neighbor who most recently contributed to lowering $\mu_{i,t}(\theta)$. The intuition here is simple: an uninformative agent cannot outdo the source of its information. This idea is made precise in claim (ii) of Proposition \ref{prop:sparsity}: {on almost all sample paths, there exists a (sample-path dependent) time $T_2(\omega)$, such that at each event-monitoring time-step $t_k > T_2(\omega)$, agent $i$ never transmits  $\mu_{i,t_k}(\theta)$ to \textit{all} its neighbors. That is, there exists at least one $j\in\mathcal{N}_i$ to which $i$ does not transmit $\mu_{i,t_k}(\theta)$.

To further demonstrate that our rule promotes sparse communication, we consider the setting described in Corollary \ref{prop:tree} where the baseline graph is a tree, and for every pair of states, there is a unique informative agent that can distinguish between them. Our result states that all upstream broadcasts to such informative agents stop after a finite period of time, almost surely. In other words, for this setting, our rule provably ensures that eventually, information flows only from informative agents to uninformative agents.} \color{black}{}
{\begin{remark}
It should be noted that the limit conditions in Eq.  \eqref{eqn:functions} are specific to Algorithm \ref{algo:ETmin}, and, as such, are only sufficient conditions for learning the true state. The condition on graph connectivity is also not necessary, and can be relaxed. However, the assumption of global identifiability is in fact necessary when agents make conditionally independent observations; see \cite{mitra2019new} for more details on this topic. 
\end{remark}}

\subsection{Asymptotic Learning of the Truth}
If asymptotic learning of the true state is all one cares about, i.e., if the convergence rate is no longer a consideration, then one can allow for arbitrarily sparse communication patterns, as we shall soon demonstrate. {In particular, our goal is to show that as long as each agent transmits its belief vector to every neighbor infinitely often, all agents will asymptotically learn the truth. We will establish the above claim as an immediate consequence of a much stronger statement that even allows the baseline network to change over time.} To this end, let $\mathcal{G}(t)=(\mathcal{V},\mathcal{E}(t))$ denote the changing neighbor graph. To allow for this general setting, we let $\mathbb{I}=\mathbb{N}_{+}$, i.e., the event condition \eqref{eq:event} is now monitored at each time-step. Furthermore, we set $\gamma(t)=\gamma \in (0,1], \forall t\in\mathbb{N}.$  At each time-step $t\in\mathbb{N}_{+}$, and for each $\theta\in\Theta$, an agent $i\in\mathcal{V}$ decides whether or not to broadcast $\mu_{i,t}(\theta)$ to an instantaneous neighbor $j\in\mathcal{N}_i(t)$ by checking the event condition \eqref{eq:event}. While checking this condition, if agent $i$ has not yet transmitted to (resp., heard from) agent $j$ about $\theta$ prior to time $t$, then it sets $\hat{\mu}_{ij,t}(\theta)$ (resp., $\hat{\mu}_{ji,t}(\theta)$) to $1$. Update rules \eqref{eq:mubardefn}, \eqref{eqn:Bayes}, \eqref{eqn:update} remain the same, with $\mathcal{N}_{i,t}(\theta)$ now interpreted as  $\mathcal{N}_{i,t}(\theta)\triangleq\{j\in\mathcal{N}_i(t)|\mathds{1}_{ji,t}(\theta)=1\}$. Finally, by an union graph over an interval $[t_1,t_2]$, we will imply the graph with vertex set $\mathcal{V}$, and edge set $\cup_{\tau=t_1}^{t_2}\mathcal{E}(\tau)$. With these modifications in place, we have the following result.
\begin{theorem}
Suppose global identifiability (Assumption \ref{assump:globiden}) holds. Furthermore, suppose for each $t\in\mathbb{N}_{+}$, the union graph over $[t,\infty)$ is rooted at $\mathcal{S}(\theta_p,\theta_q)$. Then, the event-triggered distributed learning rule described above guarantees $\mu_{i,t}(\theta^*)\rightarrow 1$ a.s. $\forall i\in\mathcal{V}.$
\label{thm:asymp} 
\end{theorem}

While a result of the above flavor is well known for the basic consensus setting \cite{moreau}, we are unaware of its analogue for the distributed inference problem. When $\mathcal{G}(t)=\mathcal{G},\forall t\in\mathbb{N}$, we observe from Theorem \ref{thm:asymp} that, as long as each agent $i$ transmits $\boldsymbol{\mu}_{i,t}$ infinitely often to each neighbor $j\in\mathcal{N}_i$, all agents will asymptotically learn the true state. In particular, other than the above requirement, our result places no constraints on the \textit{frequency} of agent interactions.

\section{A Distributed Learning Rule based on Adaptive Quantization}
\label{sec:quantrule}
The focus of Section \ref{sec:algo} was on  designing an algorithm that guarantees learning despite sparse communication. In this section, we turn our attention to promoting communication-efficiency via a complementary mechanism, namely, by compressing the amount of information transmitted by each agent. Our investigations here are motivated by the fact that in practice, communication channels modeling the interactions between agents have finite bandwidth. Accordingly, let us suppose that $\forall \theta\in\Theta$,  each agent $i$ uses only $B(\theta)$ bits to encode its belief on $\theta$.  \textit{Under what conditions on $B(\theta)$ will each agent eventually learn the true state?}

To answer the above question, we need to design an appropriate quantization scheme, which, in turn,  requires resolving the following issues. (1) The scheme should be such that the belief of each agent on $\theta^*$ converges \textit{exactly} to $1$, as opposed to getting stuck in a neighborhood of $1$. There are in fact various examples in the literature where due to quantization effects, the algorithm converges  to a neighborhood of the desired point \cite{quantopt1,quantopt2,quantopt3}. 
(2) Precaution needs to be taken to ensure that the belief of an agent on $\theta^*$ never gets quantized to $0$. Indeed, it might very well be that during an initial transient phase, the belief of some agent on $\theta^*$ falls inadvertently. If the quantization scheme is not designed appropriately, such a low belief on $\theta^*$ might get quantized to a $0$ value, causing every agent to eventually place a $0$ belief on the true state due to diffusion. This is a serious issue that needs to be addressed, and, in fact, this exact phenomenon has been reported in a simulation study conducted in \cite{lalitha}. Specifically, the authors in \cite{lalitha} present an example where using $12$ bits to represent each hypothesis leads to learning the true state, but using $8$ bits results in convergence to a false hypothesis. In what follows, we propose an algorithm that tackles the above issues; later, we argue that our algorithm guarantees exponentially fast learning even when merely $1$ bit is used to encode each hypothesis. 

To proceed, suppose we wish to encode a scalar $x$ that belongs to the interval $[L,U]$ using $B$ bit precision. Then, we first divide the interval $[L,U]$ into $2^B$ bins, each of equal width. Next, we identify the bin to which $x$ belongs, and let the quantized value of $x$ simply be the {\textit{upper end point}} of that bin. Let this entire operation be described formally by a map $\mathcal{Q}_{R,B}(\cdot)$ with range parameter $R=[L,U]$ and bit parameter $B$. Then, we have $\mathcal{Q}_{R,B}(x)=L+d\ceil{(x-L)/d}$, where $d=(U-L)/2^B$; {note here that we use the  \texttt{ceil} function for quantization.} The above encoder will serve as a basic building block for encoding each component of an agent's belief vector, and our key idea will be to sequentially refine the range of the quantizer over time.

\begin{algorithm}[t]
\caption{\textbf{(Quantized Min-Rule)}  Each agent $i \in \mathcal{V}$ executes this algorithm in parallel} \label{algo:Qmin}
\textbf{Initialization:}  $\pi_{i,0}(\theta), \mu_{i,0}(\theta)$ and $\bar{\mu}_{i,0}(\theta)$ initialized as in Algorithm \ref{algo:ETmin}; $q_{i,0}(\theta)=1,\forall \theta\in\Theta$. 
\begin{algorithmic}[1]
\For {$t \in \mathbb{N}$} 
\For {$\theta\in\Theta$}

\State Update $\pi_{i,t+1}(\theta)$ via \eqref{eqn:Bayes}, and   $\mu_{i,t+1}(\theta)$ via  \eqref{eqn:update}. 

\If{$\mu_{i,t+1}(\theta)\in[0,q_{i,t}(\theta))$}
\State Quantize $\mu_{i,t+1}(\theta)$ to $q_{i,t+1}(\theta)$ via \eqref{eqn:encoder}, and broadcast  $\mathbb{J}_{i,t+1}(\theta)$ to each $j\in\mathcal{N}_i$.
\Else
\State Set $q_{i,t+1}(\theta)=q_{i,t}(\theta)$, and do not broadcast about $\theta$.
\EndIf
\For {$j\in\mathcal{N}_i$}
\If {$j\in\mathcal{N}_{i,t+1}(\theta)$}
\State Decode $q_{j,t+1}(\theta)$ from $\mathbb{J}_{j,t+1}(\theta)$. 
\Else
\State Set $q_{j,t+1}(\theta)=q_{j,t}(\theta)$. 
\EndIf
\EndFor
\State Update $\bar{\mu}_{i,t+1}(\theta)$ via \eqref{eq:qmubar}.
\EndFor
\EndFor
\end{algorithmic}
\end{algorithm}

$\bullet$ \textbf{Encoding Beliefs:} As with Algorithm \ref{algo:ETmin}, each agent $i$ maintains a local belief vector $\boldsymbol{\pi}_{i,t}$, and an actual belief vector $\boldsymbol{\mu}_{i,t}$, which are updated via \eqref{eqn:Bayes} and \eqref{eqn:update}, respectively. In addition, for encoding its belief on $\theta$, an agent $i$ maintains a quantity $q_{i,t}(\theta)$, with $q_{i,0}(\theta)=1, \forall \theta\in\Theta$. At each time-step $t+1\in\mathbb{N}_{+}$, and for each $\theta\in\Theta$, an agent checks whether $\mu_{i,t+1}(\theta)\in[0,q_{i,t}(\theta))$. If so, it quantizes $\mu_{i,t+1}(\theta)$ to $q_{i,t+1}(\theta)=\mathcal{Q}_{R_{i,t}(\theta),B(\theta)}(\mu_{i,t+1}(\theta))$, with range parameter $R_{i,t}(\theta)=[0,q_{i,t}(\theta)]$, and a bit parameter $B(\theta)$ that will be specified later on. More precisely, if $\mu_{i,t+1}(\theta)\in[0,q_{i,t}(\theta))$, then $\mu_{i,t+1}(\theta)$ is quantized  as:
\begin{equation}
    q_{i,t+1}(\theta)=\frac{q_{i,t}(\theta)}{2^{B(\theta)}}\ceil{{\mu_{i,t+1}(\theta)2^{B(\theta)}}/{q_{i,t}(\theta)}}.
\label{eqn:encoder}
\end{equation}
Let $\mathbb{J}_{i,t+1}(\theta)$ denote the binary representation of the index of the bin to which $\mu_{i,t+1}(\theta)$ belongs. The quantized belief ${q}_{i,t+1}(\theta)$ is encoded as $\mathbb{J}_{i,t+1}(\theta)$, and the latter is broadcasted to each neighbor $j\in\mathcal{N}_i$. If $\mu_{i,t+1}(\theta) \geq q_{i,t}(\theta)$, then agent $i$ sets $q_{i,t+1}(\theta)=q_{i,t}(\theta)$, and does not broadcast about $\theta$ to any neighbor. In words, at each $t+1\in\mathbb{N}$, an agent $i$ broadcasts about $\theta$ if and only if $\mu_{i,t+1}(\theta)$ is strictly lower than the last quantized belief on $\theta$ that it broadcasted, namely $q_{i,t}(\theta)$. This last transmitted belief $q_{i,t}(\theta)$ also serves as the upper limit of the range $R_{i,t}(\theta)$ of the quantizer used for encoding $\mu_{i,t+1}(\theta)$, while the lower limit remains at $0$ for all time. The above steps constitute our \textit{adaptive quantization} scheme.\footnote{The adaptive nature of our encoding strategy stems from the fact that the range of the quantizer used to encode each hypothesis is dynamically updated.}

$\bullet$ \textbf{Decoding Beliefs:} For decoding beliefs, we make the following natural assumptions. For every $\theta\in\Theta$, each agent is aware of (i) the initial quantizer range, i.e., the fact that $q_{i,0}(\theta)=1,\forall \theta\in\Theta,\forall i\in\mathcal{V}$; (ii) the nature of the encoding operation  $\mathcal{Q}_{R,B}(\cdot)$; and (iii) the bit precision $B(\theta)$. Now consider any agent $j\in\mathcal{N}_i$. At any time-step $t+1\in\mathbb{N}_{+}$, if $j$ receives $\mathbb{J}_{i,t+1}(\theta)$ from $i$, then it can \textit{exactly} recover ${q}_{i,t+1}(\theta)$. This follows from the assumptions we made above, and the fact that node $j$ has access to $q_{i,t}(\theta)$, since it was the last quantized belief on $\theta$ that was transmitted by $i$ to each of its neighbors. If $j$ does not hear about $\theta$ from node $i$, then on its end, it sets $q_{i,t+1}(\theta)=q_{i,t}(\theta)$. 

Based on the above discussion, it should be apparent that at each time-step $t\in\mathbb{N}$, and for each $\theta\in\Theta$, the value of $q_{i,t}(\theta)$ held by an agent $i$ is consistent with those held by each of its neighbors - a fact that is crucial for correctly decoding the messages transmitted by $i$. Finally, upon completion of the decoding step, an agent $i$ updates $\bar{\mu}_{i,t+1}(\theta)$ as:
\begin{equation}
    \bar{\mu}_{i,t+1}(\theta)=\min\{\bar{\mu}_{i,t}(\theta),\mu_{i,t+1}(\theta),\{q_{j,t+1}(\theta)\}_{j\in\mathcal{N}_i}\}.
    \label{eq:qmubar}
\end{equation}

We call the above algorithm the \texttt{Quantized Min-Rule}, and outline its steps in Algorithm \ref{algo:Qmin}. In Line 10 of this algorithm, $\mathcal{N}_{i,t+1}(\theta)$ has the same meaning as in the rest of this paper: it represents the neighbors of $i$ who broadcast their beliefs (in this case, quantized beliefs) on $\theta$ to $i$ at time $t+1$. {Similar to Algorithm \ref{algo:ETmin}, implementing Algorithm \ref{algo:Qmin} imposes certain memory requirements on the part of each agent. Specifically, in addition to the vectors $\boldsymbol{\pi}_{i,t}, \boldsymbol{\mu}_{i,t}$, $\boldsymbol{\bar{\mu}}_{i,t}$, and $\boldsymbol{q}_{i,t}$, an agent $i$ needs to store the vector $\boldsymbol{q}_{j,t}$ for each neighbor $j\in\mathcal{N}_i$. The entries of $\boldsymbol{q}_{j,t}$ are the most recent quantized beliefs broadcasted by agent $j$, and storing them is necessary in order to decode the beliefs transmitted by agent $j$. Overall, for running Algorithm \ref{algo:Qmin}, each agent $i$ needs to maintain and update $(|\mathcal{N}_i|+4)$ $m$-dimensional vectors.}

{It is important to emphasize the rationale behind using a ceil operator for our quantization scheme (see Eq. \eqref{eqn:encoder}) as opposed to a floor operator. If at any point in time, the belief $\mu_{i,t}(\theta^*)$ of an agent $i$ falls in the lowest quantization bin of its quantizer range for $\theta^*$, then using a floor operator will cause $\mu_{i,t}(\theta^*)$ to get quantized to $0$. Performing a $\min$ operation  on this quantized value (as in Eq. \eqref{eq:qmubar}) will cause the output to be $0$, and eventually, via diffusion, all agents will end up with a $0$ belief on the true state $\theta^*$.  {To avoid the above phenomenon, we use a ceil operator for encoding beliefs}. Doing so ensures that the quantized belief on any hypothesis is greater than or equal to the actual belief on that hypothesis - a key component of our analysis. See also Lemma \ref{lemma:qminbound}.}

\section{Theoretical Guarantees for Algorithm \ref{algo:Qmin}}
\label{sec:quantresults}
The following is our main result concerning the convergence guarantees of Algorithm \ref{algo:Qmin}.
\begin{theorem}
Suppose every agent uses at least one bit to encode each hypothesis, i.e., let $B(\theta) \geq 1, \forall \theta\in\Theta$. Furthermore, suppose global identifiability (Assumption \ref{assump:globiden}) holds, and the communication graph $\mathcal{G}$ is connected. Then, Algorithm \ref{algo:Qmin} guarantees the following. 
\begin{itemize}
    \item \textbf{(Consistency)}: For each agent  $i\in\mathcal{V}$,  $\mu_{i,t}(\theta^{\star}) \rightarrow 1$ a.s.
    \item \textbf{(Exponentially Fast Rejection of False Hypotheses)}: For each agent  $i\in\mathcal{V}$, and for each false hypothesis $\theta\in\Theta\setminus\{\theta^{\star}\},$ the following holds:
    \begin{equation}
        \liminf_{t\to\infty}-\frac{\log\mu_{i,t}(\theta)}{t} \geq \max_{v\in\mathcal{S}(\theta^{\star},\theta)} H_v(\theta^*,\theta) \hspace{1mm} a.s., 
\label{eqn:quantizedrate}
    \end{equation}
where $H_v(\theta^*,\theta)=\min\{B(\theta) \log 2, K_v(\theta^{\star},\theta)\}$.
\end{itemize}
\label{thm:Quantmin}
\end{theorem}
We prove the above result in Appendix \ref{app:ProofThmQmin}. Under what conditions on $B(\theta)$ can one recover the same long-run learning rate as with infinite precision? The following result, which is an immediate corollary of Theorem \ref{thm:Quantmin}, provides an answer. 

\begin{corollary}
Suppose the conditions in Theorem \ref{thm:Quantmin} hold. Moreover, for each $\theta\in\Theta$, suppose the bit precision $B(\theta)$ is chosen such that
\begin{equation}
    B(\theta) \geq \frac{1}{\log2} \left(\max_{\theta^*\neq\theta} \max_{i\in\mathcal{V}} K_i(\theta^*,\theta)\right).
\label{eqn:bitprecision}
\end{equation}
Then, for each $\theta\in\Theta\setminus\{\theta^{\star}\}$, and $i\in\mathcal{V}$, we have:
    \begin{equation}
        \liminf_{t\to\infty}-\frac{\log\mu_{i,t}(\theta)}{t} \geq \max_{v\in\mathcal{S}(\theta^{\star},\theta)}K_v(\theta^{\star},\theta) \hspace{1mm} a.s.
    \end{equation}
\label{eqn:corrqmin}
\end{corollary}

 We now remark on the implications of the above results. 

$\bullet$ \textbf{1-bit precision per hypothesis is sufficient for learning:} Under standard assumptions on the observation model and the network structure, Theorem \ref{thm:Quantmin} reveals that based on Algorithm \ref{algo:Qmin}, it is possible to learn the true state exponentially fast while using just $1$ bit to encode each hypothesis. Thus, at any given  time-step, it suffices for each agent to broadcast an $m$-bit binary vector, where $m$ is the number of hypotheses. This is a key implication of Theorem \ref{thm:Quantmin}. 

{In order for each agent to learn the true state asymptotically, we conjecture that each agent must \textit{necessarily} use at least 1 bit precision to encode each hypothesis. As future work, it would be interesting to either prove or disprove this conjecture}. 

$\bullet$ \textbf{Trade-offs between bit-precision and learning rate:} While 1-bit precision per hypothesis is adequate for exponentially fast learning, the rate of learning may no longer be that with infinite precision. To understand this better, recall that with infinite precision, the basic min-rule in \cite{mitra2019new} allows each agent to rule out a false hypothesis $\theta$ exponentially fast at the rate $\max_{i\in\mathcal{V}}K_i(\theta^*,\theta)$.\footnote{Observe that setting $B(\theta)=\infty$ in \eqref{eqn:quantizedrate} leads to the same conclusion.} Let $v\in \argmax_{i\in\mathcal{V}}K_i(\theta^*,\theta)$. Although agent $v$'s belief on $\theta$ may decay to zero relatively fast, its ability to convey such a low belief to its neighbors is limited by the precision of the quantizer, when beliefs can no longer be transmitted perfectly. In particular, observe that the R.H.S. of \eqref{eqn:quantizedrate} simplifies to $\min\{B(\theta) \log 2, \max_{i\in\mathcal{S}(\theta^*,\theta)} K_i(\theta^{\star},\theta)\}$. This suggests that one can recover the same rate of rejection of $\theta$ as with infinite precision if and only if $B(\theta)\log2 \geq \max_{i\in\mathcal{S}(\theta^*,\theta)} K_i(\theta^{\star},\theta)$, i.e., a low bit-precision can come at the expense of a reduced learning rate. To sum up,  just as Theorem \ref{thm:main} highlighted the trade-offs between sparse communication and the learning rate under Algorithm \ref{algo:ETmin}, Theorem \ref{thm:Quantmin} quantifies the trade-offs between imprecise communication and the learning rate under Algorithm \ref{algo:Qmin}. 

$\bullet$ \textbf{Recovering the same learning rate as with perfect communication:} Intuitively, the condition in Eq. \eqref{eqn:bitprecision} can be interpreted as follows. To be able to reject $\theta\neq\theta^*$ at the same rate as with perfect communication, the range of the quantizer used to encode $\theta$ must shrink at least as fast as the fastest possible rate at which an agent can reject $\theta$ on its own, while accounting for the realization of any state $\theta^*\neq \theta$. However, in order to pick $B(\theta)$ to satisfy the condition in Eq.  \eqref{eqn:bitprecision}, an agent requires certain knowledge of the relative entropies of other agents in the network - this additional knowledge is the price to be paid for maintaining the same learning rate as with perfect communication (under the proposed scheme).

{
\section{Learning under Sparse and Imprecise Communication}
\label{sec:quant_event}
In sections \ref{sec:algo} and \ref{sec:quantrule}, we separately treated the aspects of event-triggering to achieve sparse communication, and adaptive quantization to deal with finite-precision channels. In this section, we will develop a learning rule that combines these ideas in a natural way. Let us begin by describing the main components of this rule. First, define a periodic sequence $\mathbb{I}=\{t_1,t_2,t_3,\ldots\} \in \mathbb{N}_{+}$ of \textit{event-monitoring} time-steps, with $t_1=1$, and period equal to a positive integer $\tau$, i.e., $t_{k+1}-t_{k} = \tau, \forall k \in \mathbb{N}_{+}$.  {We will comment on the restriction to periodic sequences later in the section.} Next, as in Algorithm \ref{algo:ETmin}, we consider a non-decreasing \textit{threshold} function $\gamma:\mathbb{N} \rightarrow (0,1]$. Finally, we consider the same encoder map $\mathcal{Q}_{R,B}(\cdot)$ as described in Section \ref{sec:quantrule}. 

$\bullet$ \textbf{Algorithm Description:} Each agent $i$ maintains the vectors $\boldsymbol{\pi}_{i,t}$, $\boldsymbol{\mu}_{i,t}$, $\boldsymbol{\bar{\mu}}_{i,t}$, $\boldsymbol{q}_{i,t}$, and \{$\boldsymbol{q}_{j,t}\}_{j\in\mathcal{N}_i}$; these vectors have exactly the same meaning as in Algorithm \ref{algo:Qmin}, and are initialized in the same way. At each time-step $t+1\in\mathbb{N}_{+}$, and for each hypothesis $\theta\in\Theta$, the following steps are executed by each agent in parallel.  (i) Agent $i$ updates $\pi_{i,t+1}(\theta)$ via \eqref{eqn:Bayes} and $\mu_{i,t+1}(\theta)$ via \eqref{eqn:update}. (ii) If $t+1 \in \mathbb{I}$, agent $i$ checks the following event condition:
\begin{equation}
    \mu_{i,t+1}(\theta) < \gamma(t+1) q_{i,t}(\theta).
\label{eqn:event_quant}
\end{equation}
If the above condition holds, then $\mu_{i,t+1}(\theta)$ is quantized to $q_{i,t+1}(\theta)$ based on equation \eqref{eqn:encoder}. Agent $i$ then encodes $q_{i,t+1}(\theta)$ as $\mathbb{J}_{i,t+1}(\theta)$ -- the binary representation of the index of the bin to which $\mu_{i,t+1}(\theta)$ belongs. {The index $\mathbb{J}_{i,t+1}(\theta)$ is then transmitted to every neighbor in $\mathcal{N}_i$.} If the event condition in Eq. \eqref{eqn:event_quant} fails, then agent $i$ sets $q_{i,t+1}(\theta)=q_{i,t}(\theta)$, and does not broadcast about $\theta$ to any neighbor. (iii) If $t+1 \notin \mathbb{I}$, then agent $i$ does not communicate with any neighbor, and sets $q_{i,t+1}(\theta)=q_{i,t}(\theta)$. (iv) Beliefs are decoded exactly as in Algorithm \ref{algo:Qmin}, and $\bar{\mu}_{i,t+1}(\theta)$ is updated based on \eqref{eq:qmubar}. We call the above algorithm the \texttt{Quantized Event-Triggered Min-Rule}, or simply, the  \texttt{QET Min-Rule}. 

A few points are worth highlighting about the above algorithm. First, note that unlike Algorithm \ref{algo:Qmin}, updates to an agent's quantizer, and all inter-agent interactions, are restricted to time-steps in $\mathbb{I}$ subject to the event condition \eqref{eqn:event_quant}. Between two consecutive event-monitoring time-steps, the quantizers maintained by each agent remain unchanged, and there is no communication between agents. Unlike the event condition \eqref{eq:event}, the one in \eqref{eqn:event_quant} checks whether an agent's current belief on $\theta$ has fallen significantly below the last \textit{quantized} belief on $\theta$ it broadcasted. One could, in principle, design a more involved \textit{agent-specific} event condition (in the spirit of \eqref{eq:event}) that also incorporates feedback from the neighbors. However, this would require an agent $i$ to maintain specific quantizers for each of its neighbors, significantly complicating the design and analysis of the resulting algorithm. We do not investigate such a complex mechanism here since our main aim is to (i) highlight how one can, in a simple way,  blend the ideas of event-triggering and quantization; (ii) provide a sense for the flavor of results one can expect when these ideas are combined. With this in mind, we now state the main result of this section; for its proof, see Appendix \ref{app:event_quant}. 

\begin{theorem}
Suppose every agent uses at least one bit to encode each hypothesis, i.e., let $B(\theta) \geq 1, \forall \theta\in\Theta$, and let $\bar{B}(\theta)=B^{1/\tau}(\theta)$, where $\tau$ is the communication period. Let the threshold function $\gamma(\cdot)$ satisfy the condition in \eqref{eqn:functions}. Furthermore, suppose global identifiability (Assumption \ref{assump:globiden}) holds, and the communication graph $\mathcal{G}$ is connected. Then, the \texttt{QET Min-rule} guarantees consistency, i.e., for each agent  $i\in\mathcal{V}$,  $\mu_{i,t}(\theta^{\star}) \rightarrow 1$ almost surely. Moreover, for each $i\in\mathcal{V}$, and for each  $\theta\in\Theta\setminus\{\theta^{\star}\},$ the following holds:
    \begin{equation}
        \liminf_{t\to\infty}-\frac{\log\mu_{i,t}(\theta)}{t} \geq \max_{v\in\mathcal{S}(\theta^{\star},\theta)} \bar{H}_v(\theta^*,\theta) \hspace{1mm} a.s., 
\label{eqn:quant_event_rate}
    \end{equation}
where $\bar{H}_v(\theta^*,\theta)=\min\{\bar{B}(\theta) \log 2, K_v(\theta^{\star},\theta)\}$.
\label{thm:event_quant}
\end{theorem}

\textbf{Discussion:} From Theorem \ref{thm:event_quant}, we note that as long as the threshold function $\gamma(\cdot)$ does not decay too fast, we essentially end up getting similar guarantees as in Theorem \ref{thm:Quantmin}. However, the key distinction from Theorem \ref{thm:Quantmin} lies in the effect of the communication period $\tau$ on the asymptotic learning rate: whereas we had a $B(\theta)\log(2)$ term showing up in the rate of learning earlier (see Eq. \eqref{eqn:quantizedrate}), we now have a $B^{1/\tau}(\theta)\log(2)$ term taking its place instead. It is instructive to compare this result with that of Corollary \ref{corr:functionalform} where we saw that, in the absence of quantization, even when the gap between successive event-monitoring steps grows polynomially, the long-term learning rate remains unaffected. In contrast, Theorem \ref{thm:event_quant} tells us that even a constant gap of $\tau$ does impact the convergence rate of the \texttt{QET Min-rule}, {suggesting that growing event-interval functions can significantly slow down the convergence rate}. This phenomenon can be essentially attributed to the fact that between successive event-monitoring time-steps, the quantizers at any given agent are never updated. If they were, the neighbors of this agent would not be aware of such updates, and hence, would perform incorrect decoding based on stale information. {To keep the analysis simple, and at the same time provide the above insights, we considered periodic communication patterns in this section.}}
\color{black}{}

\section{Simulations}
\label{sec:simulations}
\begin{figure}[t]
\hspace{-3mm}
\begin{tabular}{cc}
\includegraphics[scale=0.295]{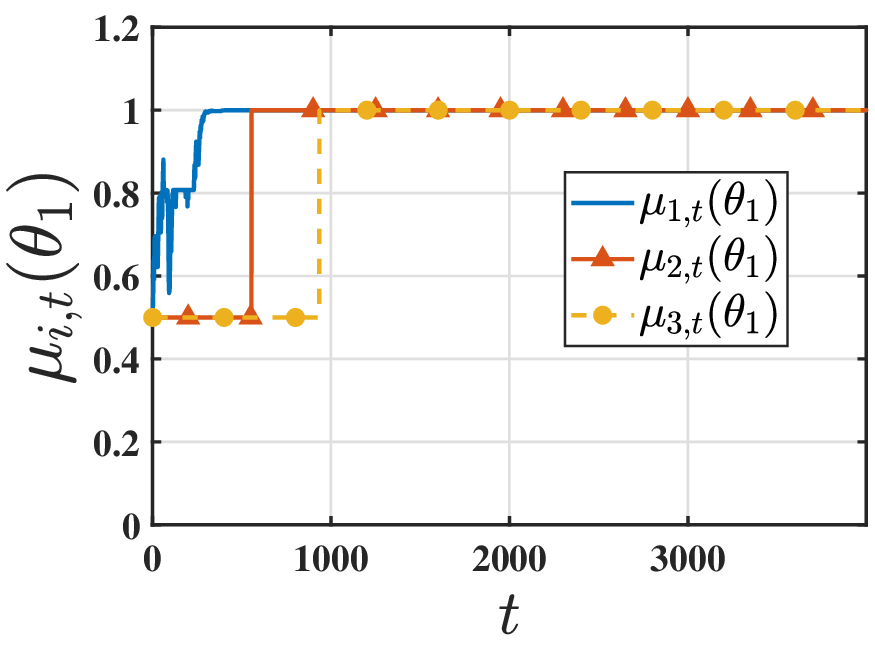}&\hspace{-3mm}\includegraphics[scale=0.295]{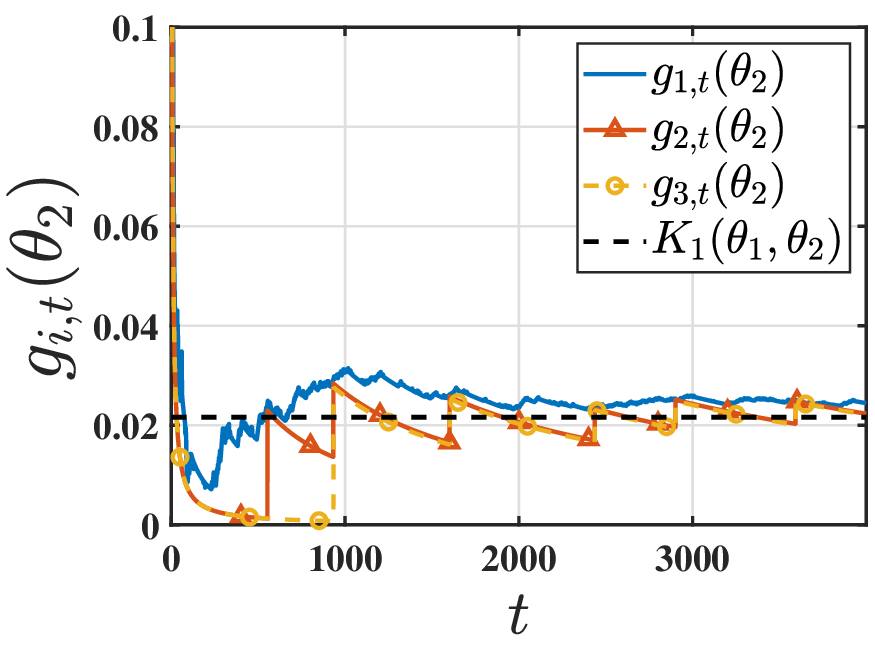}\\
(a) & (b)\\
\includegraphics[scale=0.295]{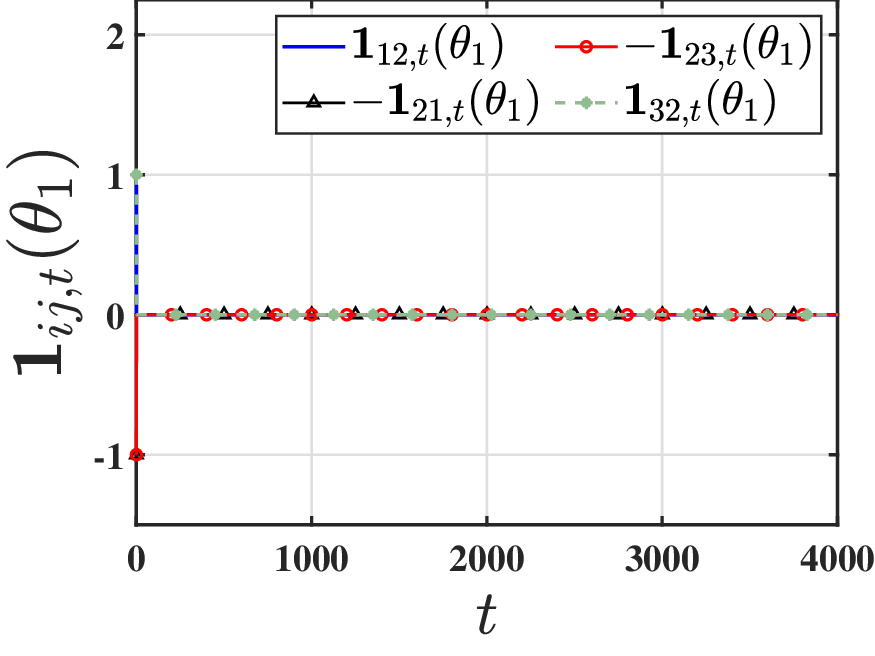}&\hspace{-3mm} \includegraphics[scale=0.295]{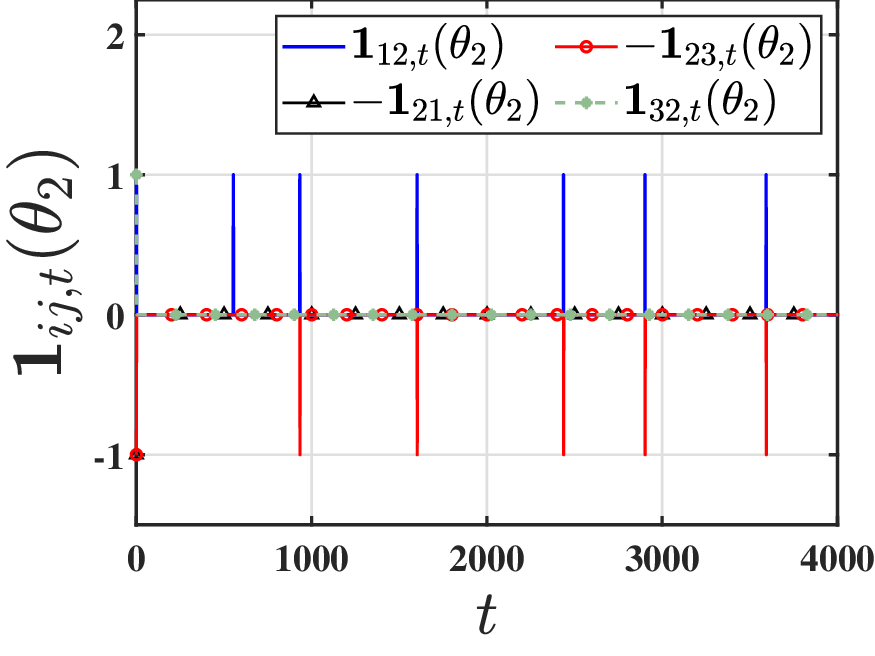}\\
(c) & (d)
\end{tabular}
\caption{Plots pertaining to the simulation example in Sec. \ref{sec:simulations}. Fig. \ref{fig:plots}(a) plots the belief evolutions on the true state $\theta_1$. Fig. \ref{fig:plots}(b) plots the rate at which each agent rejects the false hypothesis $\theta_2$, namely $q_{i,t}(\theta_2)=-\log(\mu_{i,t}(\theta_2))/t$. Fig.'s \ref{fig:plots}(c) and \ref{fig:plots}(d) demonstrate the sparse communication patterns generated by our event-triggering scheme.} 
\label{fig:plots}
\end{figure}
To validate our key theoretical findings, we first consider the simple 3-agent network in Fig. \ref{fig:example}. Suppose $\Theta=\{\theta_1,\theta_2\}, \theta^*=\theta_1$, and let the signal space for each agent be $\{0,1\}$. The likelihood models are as follows: $l_1(0|\theta_1)=0.7, l_1(0|\theta_2)=0.6$, and $l_i(0|\theta_1)=l_i(0|\theta_2)=0.5, \forall i\in\{2,3\}$. Clearly, agent $1$ is the only informative agent. To isolate the impact of our event-triggering strategy, we set $g(k)=1, \forall k\in\mathbb{N}_{+}$, i.e., the event condition in Eq. \eqref{eq:event} is monitored at every time-step. We set the threshold function as $\gamma(k)=1/k^2$. The performance of Algorithm \ref{algo:ETmin} is depicted in Fig. \ref{fig:plots}. We make the following observations. (i) From Fig. \ref{fig:plots}(a), we note that all agents eventually learn the truth. (ii) From Fig. \ref{fig:plots}(b), we note that the asymptotic rate of rejection of the false hypothesis $\theta_2$, namely ${g_{i,t}(\theta_2)}=-\log(\mu_{i,t}(\theta_2))/t$, complies with the theoretical bound in Thm.  \ref{thm:main}. (iii) From Fig. \ref{fig:plots}(c), we note that after the first time-step, all agents stop broadcasting about the true state $\theta_1$, complying with claim (i) of Prop. \ref{prop:sparsity}. (iv) From Fig. \ref{fig:plots}(d), we note that broadcasts about $\theta_2$ along the path $3\rightarrow2\rightarrow1$ stop after the first time-step, in accordance with claim (ii) of Prop. \ref{prop:sparsity}, and Corr. \ref{prop:tree}. We also observe that in the first 4000 time-steps, agent 1 (resp., agent 2) broadcasts its belief on $\theta_2$ to agent 2 (resp., agent 3) only 7 times (resp., 6 times). Despite such drastic reduction in the number of communication rounds, all agents still learn the truth with no loss in learning rate relative to the baseline algorithm in \cite{mitra2019new}. This demonstrates the effectiveness of our approach.

\begin{figure}[t]
\hspace{-4mm}
\begin{tabular}{cc}
\includegraphics[scale=0.295]{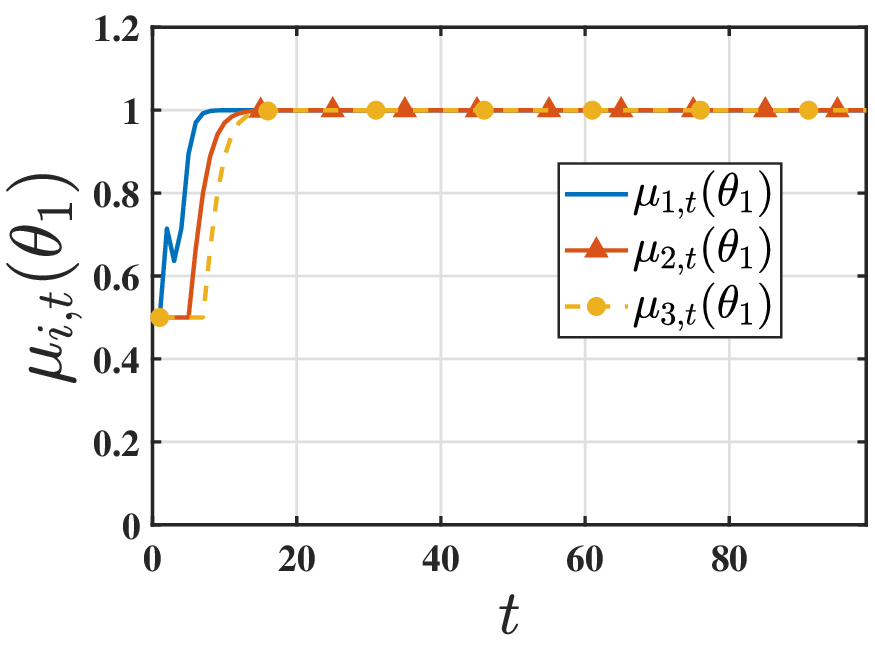}& \hspace{-3mm}  \includegraphics[scale=0.295]{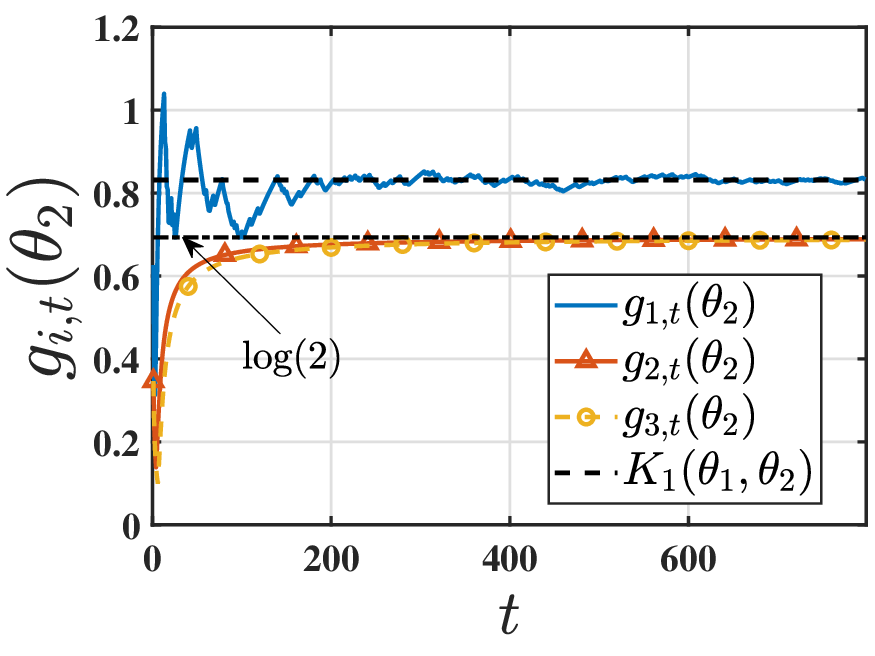}\\
(a) & (b)\\
\end{tabular}
\caption{Plots concerning the performance of Algorithm \ref{algo:Qmin} for the network in Fig \ref{fig:example}, when $1$ bit is used to encode each hypothesis. Figures \ref{fig:qplots}(a) and \ref{fig:qplots}(b) are analogous to Figures  \ref{fig:plots}(a) and \ref{fig:plots}(b). These plots demonstrate that while learning is possible even with $1$-bit precision, the learning rate exhibits a dependence on the quantizer precision level.}
\label{fig:qplots}
\end{figure}

As our second simulation study, we investigate the performance of our quantized learning rule, namely Algorithm \ref{algo:Qmin}. To do so, keeping everything else the same, suppose we now modify the likelihood model of agent 1 as follows: {$l_1(0 | \theta_1)=0.8$ and $l_1(0|\theta_2)=0.2$}. Fig.  \ref{fig:qplots} depicts the performance of Algorithm \ref{algo:Qmin} for this scenario, when $B(\theta_1)=B(\theta_2)=1$, i.e., when $1$ bit is used to encode each hypothesis. From Fig. \ref{fig:qplots}(a), we note that all agents learn the true state. Fig. \ref{fig:qplots}(b) reveals that the learning rates of the uninformative agents 2 and 3 are limited by the precision of the quantizer. In particular, since $K_1(\theta_1,\theta_2)=0.8318 > \log(2)$, the learning rates for these agents get saturated at $\log(2)$, exactly as suggested by Eq. \eqref{eqn:quantizedrate} in Theorem \ref{thm:Quantmin}. Despite these quantization effects, we observe that the beliefs of all agents converge to $\theta^*$ quite fast.

\begin{figure}[t]
\hspace{-4mm}
\begin{tabular}{cc}
\includegraphics[scale=0.295]{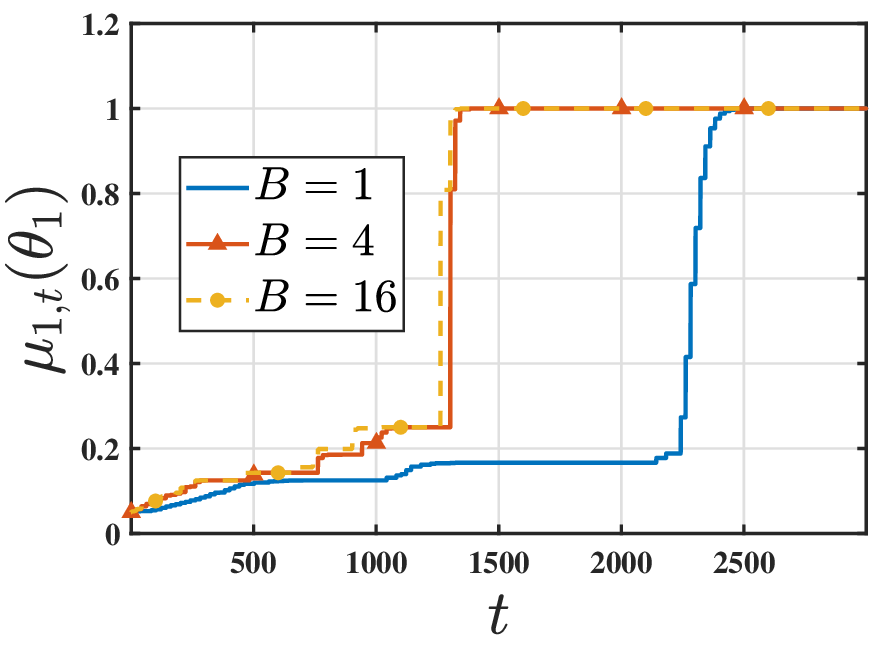}& \hspace{-3mm}  \includegraphics[scale=0.295]{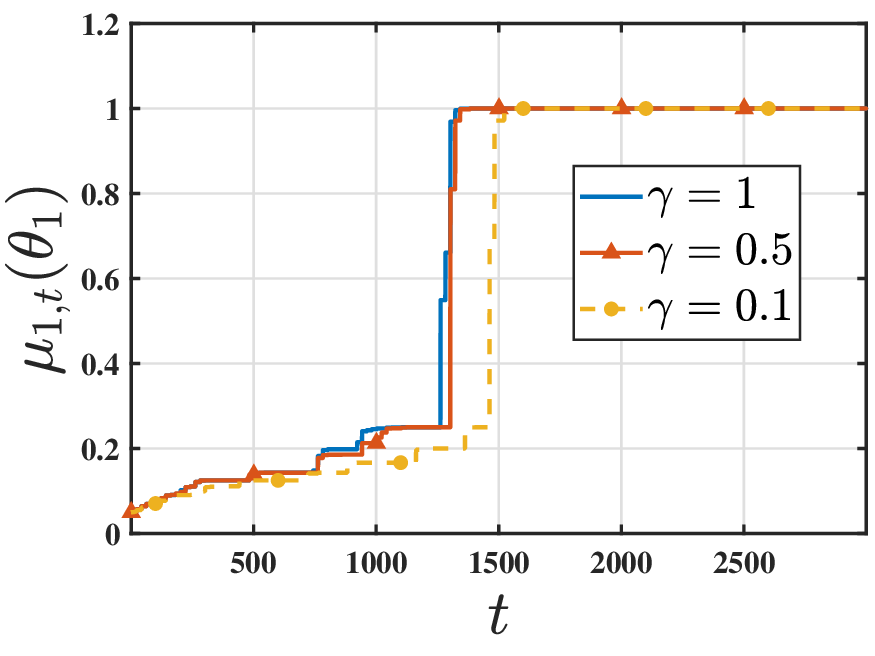}\\
(a) & (b)\\
\end{tabular}
\caption{{Plots concerning the performance of the \texttt{QET} Min-Rule on a 100-agent ring graph. Figures \ref{fig:QETplot}(a) and \ref{fig:QETplot}(b) show the evolution of agent 1's belief on $\theta^*=\theta_1$ with varying bit-precision, and varying threshold $\gamma(\cdot)$, respectively.}}
\label{fig:QETplot}
\end{figure}
{
To show that our framework extends to larger networks, we evaluate the performance of the \texttt{QET} Min-Rule of Section \ref{sec:quant_event} on a 100-agent ring graph. For this example, we let $m=20$, i.e., there are $20$ hypotheses. The common signal space is still $\{0,1\}$. The agents' likelihood models are generated as follows. For each agent $i$, we first draw an index $r_i$ uniformly at random from the set \{1,\ldots,20\}, and then set $l_i(0|\theta_{r_i})=0.7$, and $l_i(0|\theta)=0.5,  \forall \theta \in \Theta\setminus\{\theta_{r_i}\}$. Thus, agent $i$ is informative about $\theta_{r_i}$ in the sense that it can distinguish $\theta_{r_i}$ from every other hypothesis. However, agent $i$'s observation model is uninformative w.r.t. every other hypothesis. We let $\theta^*=\theta_1$. Based on the randomly generated likelihood models, only agents $28,42,79,$ and $82$ can identify $\theta_1$ as the true hypothesis on their own; the remaining agents are thus reliant on information diffusion for identifying the truth. With the communication period $\tau$ set to $20$, we plot agent 1's belief evolution on the true state $\theta_1$ in Fig. \ref{fig:QETplot}. In Fig. \ref{fig:QETplot}(a), we study how the bit-precision level impacts convergence time: as one would expect, using more bits leads to faster convergence to the truth. Nonetheless, even for this large 100-agent network with connectivity equal to just 2, Fig. \ref{fig:QETplot}(a) reveals that 1-bit precision suffices for learning. In Fig \ref{fig:QETplot}(b), we see that using a smaller threshold $\gamma(\cdot)$ slows down convergence, aligning with intuition; for this experiment, we fix the bit precision to $4$ bits.} 

\color{black}{}

\section{Conclusion}
We developed novel  learning algorithms to solve the distributed inference problem in the face of sparse and imprecise communication. For reducing the  communication frequency, we proposed an event-triggered rule that has the potential to significantly limit information flow from uninformative agents to informative agents. To deal with finite bandwidth constraints, we developed a learning rule based on adaptive quantization that allows each agent to learn the true state exponentially fast using just 1 bit to encode each hypothesis. Finally, we showed how the ideas of event-triggering and adaptive quantization can be effectively combined. Our analysis provides several new insights into the trade-offs between communication-efficiency and the learning rate. As future work, we plan to undertake a finer non-asymptotic analysis of our algorithms to reveal trade-offs between communication-efficiency and transient performance. We also plan to explore more general settings where the unknown parameter is no longer restricted to a finite set. 

\appendices
\section{Proofs pertaining to Section \ref{sec:results}}
\label{sec:proofs}
In this section, we provide proofs of all the results stated in Section \ref{sec:results}. We begin by compiling various useful properties of our update rule which will be useful later on.
\begin{lemma}
 Suppose the conditions in Theorem \ref{thm:main} hold. Then, there exists a set $\bar{\Omega}\subseteq\Omega$ with the following properties. (i) $\mathbb{P}^{\theta^{\star}}(\bar{\Omega})=1$. (ii) For each $\omega\in\bar{\Omega}$, there exist constants $\eta(\omega)\in(0,1)$ and $t'(\omega)\in(0,\infty)$ such that
\begin{equation}
\pi_{i,t}(\theta^{\star}) \geq \eta(\omega), \bar{\mu}_{i,t}(\theta^{\star}) \geq \eta(\omega), \forall t \geq t'(\omega),\forall i\in\mathcal{V}.
\label{eqn:lowerbound}
    \end{equation}
(iii) Consider a false hypothesis $\theta\neq \theta^*$, and an agent $i\in\mathcal{S}(\theta^*,\theta)$. Then on each sample path $\omega\in\bar{\Omega}$, we have:
\begin{equation}
        \liminf_{t\to\infty}-\frac{\log\mu_{i,t}(\theta)}{t} \geq K_i(\theta^{\star},\theta) \hspace{0.25mm}. 
        \label{eqn:source_rate}
    \end{equation}
\label{lemma:bound}
\end{lemma}
\begin{proof}
The proof of claim (ii) rests on the same ideas as that of \cite[Lemma 2]{mitra2019new}; we thus only sketch the main arguments for completeness. From \cite[Lemma 2]{mitra2019new}, there exists a set $\bar{\Omega}\subseteq\Omega$ with $\mathbb{P}^{\theta^{\star}}(\bar{\Omega})=1$ such that for each $\omega\in\bar{\Omega}$, the following are true for every $i\in\mathcal{V}$: (i) $\pi_{i,t}(\theta^*)>0, \forall t\in\mathbb{N}$; and (ii) $\exists \delta >0$ and $t'(\omega) < \infty$ such that $\pi_{i,t}(\theta^*) \geq \delta, \forall t\geq t'(\omega)$. Fix an $\omega\in\bar{\Omega}$. Let $\rho(\omega)=\min_{i\in\mathcal{V}}\{\bar{\mu}_{i,t'(\omega)-1}(\theta^*)\}$. Based on update rules \eqref{eq:mubardefn} and \eqref{eqn:update}, observe that $\rho(\omega) >0$; for if not, this would necessarily imply that $\pi_{i,t}(\theta^*)=0$ for some agent $i$ at some time-step $t\leq t'(\omega)-1$, which would be a contradiction given our choice of $\omega$. Let $\eta(\omega)=\min\{\delta,\rho(\omega)\}$, fix an agent $i$, and consider the update of $\mu_{i,t'(\omega)}(\theta^*)$ based on \eqref{eqn:update}:
\begin{equation}
\begin{aligned}
\mu_{i,t'(\omega)}(\theta^*)&=\frac{\min\{\bar{\mu}_{i,t'(\omega)-1}(\theta^*),\pi_{i,t'(\omega)}(\theta^*)\}}{\sum\limits_{p=1}^{m}\min\{\bar{\mu}_{i,t'(\omega)-1}(\theta_p),\pi_{i,t'(\omega)}(\theta_p)\}}\\
&\geq \frac{\eta(\omega)}{\sum\limits_{p=1}^{m}\pi_{i,t'(\omega)}(\theta_p)}=\eta(\omega),
\end{aligned}
\end{equation}
where the last equality follows from the fact that the local belief vectors generated via \eqref{eqn:Bayes} are valid probability distributions over $\Theta$ at each time-step, and hence $\sum\limits_{p=1}^{m}\pi_{i,{t}'(\omega)}(\theta_p)=1$. The above argument applies identically to every agent in the graph, and hence we have from $\eqref{eq:mubardefn}$ that $\bar{\mu}_{i,t'(\omega)}(\theta^*)=\min\{\bar{\mu}_{i,t'(\omega)-1}(\theta^*),\{\mu_{j,t'(\omega)}(\theta^*)\}_{j\in \{i\}\cup \mathcal{N}_{i,t'(\omega)}(\theta^*)}\} \geq \eta(\omega)$. We have thus argued that for every agent  $i\in\mathcal{V}$, $\mu_{i,t'(\omega)}(\theta^*) \geq \eta(\omega), \bar{\mu}_{i,t'(\omega)}(\theta^*) \geq \eta(\omega)$. We can keep repeating the above analysis for each $t > t'(\omega)$ to establish \eqref{eqn:lowerbound}. Claim (iii) in Lemma \ref{lemma:bound} follows the same reasoning as  \cite[Lemma 3]{mitra2019new}.
\end{proof}

The above lemma informs us that the belief $\mu_{v,t}(\theta)$ of an agent $v\in\mathcal{S}(\theta^*,\theta)$ decays exponentially fast at a rate lower-bounded by $K_v(\theta^*,\theta)$ on a set of $\mathbb{P}^{\theta^*}$-measure 1. How does this impact the belief $\mu_{i,t}(\theta)$ of an agent $i\in\mathcal{V}\setminus\mathcal{S}(\theta^*,\theta)$? The following result answers this question. 

\begin{lemma}
Consider a false hypothesis $\theta\in\Theta\setminus\{\theta^{\star}\}$ and an agent $v\in\mathcal{S}(\theta^{\star},\theta)$. Suppose the conditions stated in Theorem \ref{thm:main} hold. Then, the following is true for each agent $i\in\mathcal{V}$:
\begin{equation}
 \liminf\limits_{t\to\infty}-\frac{\log\mu_{i,t}(\theta)}{t} \geq {\alpha^{d(v,i)}}K_v(\theta^{\star},\theta) \hspace{1mm} a.s.
      \label{eqn:eachsource}
    \end{equation}
\label{lemma:main}
\end{lemma}
\begin{proof}
Let $\bar{\Omega}\subseteq \Omega$ be the set of sample paths for which assertions (i)-(iii) of Lemma \ref{lemma:bound} hold. Fix a sample path  $\omega\in\bar{\Omega}$, an agent $v\in\mathcal{S}(\theta^{\star},\theta)$, and an agent $i\in\mathcal{V}$. When $i=v$, the assertion of Eq. \eqref{eqn:eachsource} follows directly from Eq. \eqref{eqn:source_rate} in Lemma \ref{lemma:bound}. In particular, this implies that for a fixed $\epsilon > 0$,  $\exists {t}_v(\omega,\theta,\epsilon)$ such that:
\begin{equation}
    \mu_{v,t}(\theta) <  e^{-(K_v(\theta^{\star},\theta)-\epsilon)t}, \forall t \geq {t}_v(\omega,\theta,\epsilon).
    \label{eqn:bound1}
\end{equation}
Moreover, since $\omega\in\bar{\Omega}$,  Lemma \ref{lemma:bound} guarantees the existence of a time-step $t'(\omega) < \infty$, and a constant $\eta(\omega)>0$, such that on $\omega$, $\pi_{i,t}(\theta^{\star}) \geq \eta(\omega), \bar{\mu}_{i,t}(\theta^{\star}) \geq \eta(\omega), \forall t\geq t'(\omega), \forall i\in\mathcal{V}$. Let $\bar{t}_v(\omega,\theta,\epsilon)=\max\{t'(\omega),t_v(\omega,\theta,\epsilon)\}$.  Let $t_q > \bar{t}_v$ be the first event-monitoring time-step in $\mathbb{I}$ that is larger than  $\bar{t}_v$.\footnote{We will henceforth suppress the dependence of various quantities on $\omega,\theta$, and $\epsilon$ for brevity.} Now consider any $t_k\in\mathbb{I}$ such that $t_k\geq t_q$. In what follows, we will analyze the implications of agent $v$ deciding whether or not to broadcast its belief on $\theta$ to a one-hop neighbor $j\in\mathcal{N}_v$ at $t_k$. To this end, we consider the following two cases. 

\underline{\textbf{Case 1}: $\mathds{1} _{vj,t_k}(\theta)=1$}, i.e., $v$ broadcasts $\mu_{v,t_k}(\theta)$ to $j$ at $t_k$. Thus, since $v\in\mathcal{N}_{j,t_k}(\theta)$, we have $\bar{\mu}_{j,t_k}(\theta)\leq \mu_{v,t_k}(\theta)$ from \eqref{eq:mubardefn}. Let us now observe that $\forall t\geq t_{k}+1$: 
\begin{equation}
\resizebox{0.9\hsize}{!}{$
\begin{aligned}
\mu_{j,t}(\theta)&\overset{(a)}{\leq}\frac{\bar{\mu}_{j,t-1}(\theta)}{\sum\limits_{p=1}^{m}\min\{\bar{\mu}_{j,t-1}(\theta_p),\pi_{j,t}(\theta_p)\}}\\
&\overset{(b)}{\leq}\frac{\mu_{v,t_k}(\theta)}{\sum\limits_{p=1}^{m}\min\{\bar{\mu}_{j,t-1}(\theta_p),\pi_{j,t}(\theta_p)\}} 
\overset{(c)}{<}\frac{e^{-(K_v(\theta^{\star},\theta)-\epsilon)t_k}}{\eta}.
\end{aligned}
$}
\label{eqn:case1bnd}
\end{equation}
In the above inequalities, (a) follows directly from \eqref{eqn:update}, (b) follows by noting that the sequence $\{\bar{\mu}_{j,t}(\theta)\}$ is non-increasing based on \eqref{eq:mubardefn}, and (c) follows from \eqref{eqn:bound1} and the fact that all beliefs on $\theta^{\star}$ are bounded below by $\eta$ for $t \geq \bar{t}_v$. 

\underline{\textbf{Case 2}: $\mathds{1} _{vj,t_k}(\theta) \neq 1$}, i.e., $v$ does not broadcast $\mu_{v,t_k}(\theta)$ to $j$ at $t_k$. From the event condition in \eqref{eq:event}, it must then be that at least one of the following is true: (a) $\mu_{v,t_k}(\theta) \geq \gamma(t_k) \hat{\mu}_{vj,t_k}(\theta)$, and (b) $\mu_{v,t_k}(\theta) \geq \gamma(t_k) \hat{\mu}_{jv,t_k}(\theta)$. Suppose $\mu_{v,t_k}(\theta) \geq \gamma(t_k) \hat{\mu}_{vj,t_k}(\theta)$. From \eqref{eqn:bound1}, we then have:
\begin{equation}
    \hat{\mu}_{vj,t_k}(\theta)\leq \frac{\mu_{v,t_k}(\theta)} {\gamma(t_k)} <  \frac{e^{-(K_v(\theta^{\star},\theta)-\epsilon)t_k}}{\gamma(t_k)}.
    \label{eqn:case2(a)interim}
\end{equation}
In words, the above inequality places an upper bound on the belief of agent $v$ on $\theta$ when it last transmitted its belief on $\theta$ to agent $j$, \textit{prior} to time-step $t_k$; at least one such transmission is guaranteed to take place since all agents broadcast their entire belief vectors to their neighbors at $t_1$. Noting that $\bar{\mu}_{j,t}(\theta) \leq  \hat{\mu}_{vj,t_k}(\theta)$, $\forall t\geq t_k$, using \eqref{eqn:update},  \eqref{eqn:case2(a)interim}, and arguments similar to those for arriving at \eqref{eqn:case1bnd}, we obtain:  
\begin{equation}
    \mu_{j,t}(\theta) < \frac{e^{-(K_v(\theta^{\star},\theta)-\epsilon)t_k}}{\eta \gamma(t_k)} \leq \frac{e^{-(K_v(\theta^{\star},\theta)-\epsilon)t_k}}{\eta\gamma(t)}, \forall t\geq t_k+1,
    \label{eqn:case2bnd}
\end{equation}
where the last inequality follows from the fact that $\gamma(\cdot)$ is a non-increasing function of its argument. Now consider the case when $\mu_{v,t_k}(\theta) \geq \gamma(t_k) \hat{\mu}_{jv,t_k}(\theta)$. Following the same reasoning as before, we can arrive at an identical upper-bound on $\hat{\mu}_{jv,t_k}(\theta)$ as in \eqref{eqn:case2(a)interim}. Using the definition of $\hat{\mu}_{jv,t_k}(\theta)$, and the fact that agent $j$ incorporates its own belief on $\theta$ in the update rule \eqref{eq:mubardefn}, we have that $\bar{\mu}_{j,t}(\theta)\leq\hat{\mu}_{jv,t_k}(\theta), \forall t\geq t_{k}$. Using similar arguments as before, observe that the bound in \eqref{eqn:case2bnd} holds for this case too.

Combining the analyses of cases 1 and 2, referring to \eqref{eqn:case1bnd} and \eqref{eqn:case2bnd}, and noting that $\gamma(t)\in (0,1], \forall t\in\mathbb{N}$, we conclude that the bound in \eqref{eqn:case2bnd} holds for each $t_k\in\mathbb{I}$ such that $t_k > \bar{t}_v$. Now since $t_{k+1}-t_{k}=g(k)$, for any $\tau\in\mathbb{N}_{+}$ we have:
\begin{equation}
    t_{q+\tau}=t_q+\sum\limits_{z=q}^{q+\tau-1}g(z).
\end{equation}
Next, noting that $g(\cdot)$ is non-decreasing, observe that:
\begin{equation}
   t_q+ \int\limits_{q}^{q+\tau}g(z-1)dz \leq t_{q+\tau} \leq t_q+\int\limits_{q}^{q+\tau}g(z)dz. 
\end{equation}
The above yields: $l(q,\tau)\triangleq t_q+G(q+\tau-1)-G(q-1) \leq t_{q+\tau} \leq t_q+G(q+\tau)-G(q)\triangleq u(q,\tau)$. Fix any time-step $t > u(q,1)$, let $\tau(t)$ be the largest index such that $u(q,\tau(t)) < t$, and $\bar{\tau}(t)$ be the largest index such that $t_{q+\bar{\tau}(t)} < t$. Observe:
\begin{equation}
    \bar{t}_v < t_q < t_{q+1} \leq t_{q+\tau(t)} \leq t_{q+\bar{\tau}(t)} < t.
\end{equation}
Using the above inequality, the fact that $l(q,\tau(t)) \leq t_{q+\tau(t)}$, and referring to \eqref{eqn:case2bnd}, we obtain:
\begin{equation}
    \mu_{j,t}(\theta) < \frac{e^{-(K_v(\theta^{\star},\theta)-\epsilon)t_{q+\bar{\tau}(t)}}}{\eta \gamma(t)}  \leq \frac{e^{-(K_v(\theta^{\star},\theta)-\epsilon)l(q,\tau(t))}}{\eta \gamma(t)}.
\label{eqn:finalbndinterim}
\end{equation}
From the definition of $\tau(t)$ and $u(q,\tau(t))$, we have $q+\tau(t)=\ceil*{G^{-1}(t-t_q+G(q))}-1$. This yields:
\begin{equation}
\begin{aligned}
    l(q,\tau(t))&=t_q+G(\ceil*{G^{-1}(t-t_q+G(q))}-2)-G(q-1)\\
    & \geq t_q+G(G^{-1}(t-t_q+G(q))-2)-G(q-1).
    \end{aligned}
    \label{eqn:lqt}
\end{equation}
From  \eqref{eqn:finalbndinterim} and \eqref{eqn:lqt}, we obtain the following $\forall t > u(q,1)$:
\begin{equation}
    -\frac{\log\mu_{j,t}(\theta)}{t} > \frac{\tilde{G}(t)}{t}(K_v(\theta^{\star},\theta)-\epsilon) -\frac{\log c}{t} -\frac{\log(1/\gamma(t))}{t},
\label{eqn:finalbnd}
\end{equation}
where $\tilde{G}(t)=G(G^{-1}(t-t_q+G(q))-2)$, and $c=e^{-(K_v(\theta^*,\theta)-\epsilon)(t_q-G(q-1))}/\eta$. Now taking the limit inferior on both sides of \eqref{eqn:finalbnd} and using \eqref{eqn:functions} yields:
\begin{equation}
    \liminf\limits_{t\to\infty}-\frac{\log\mu_{j,t}(\theta)}{t} \geq {\alpha}(K_v(\theta^{\star},\theta)-\epsilon).
\label{eqn:onehoprate}
\end{equation}
Finally, since the above inequality holds for any sample path $\omega\in\bar{\Omega}$, and an arbitrarily small $\epsilon$, it follows that the assertion in \eqref{eqn:eachsource} is true for every one-hop neighbor $j$ of agent $v$. 

Now consider any agent $i$ such that $d(v,i)=2$. Clearly, there must exist some $j\in\mathcal{N}_v$ such that $i\in\mathcal{N}_j$. Following an identical line of reasoning as before, it is easy to see that with $\mathbb{P}^{\theta^*}$-measure 1,  $\mu_{i,t}(\theta)$ decays exponentially at a rate that is at least $\alpha$ times the rate at which $\mu_{j,t}(\theta)$ decays to zero. From \eqref{eqn:onehoprate}, the latter rate is at least $\alpha K_v(\theta^*,\theta)$, and hence, the former is at least  $\alpha^2 K_v(\theta^*,\theta)$. This establishes the claim of the lemma for all agents that are two-hops away from agent $v$. Since $\mathcal{G}$ is connected, given any $i\in\mathcal{V}$, there exists a path $\mathcal{P}(v,i)$ in $\mathcal{G}$ from $v$ to $i$. One can keep repeating the above argument along the path $\mathcal{P}(v,i)$ to complete the proof. 
\end{proof}

We are now in position to prove Theorem \ref{thm:main}.
\begin{proof} (\textbf{Theorem \ref{thm:main}}) Fix a  $\theta\in\Theta\setminus\{\theta^{\star}\}$. Based on condition (i) of the Theorem, $\mathcal{S}(\theta^{\star},\theta)$ is non-empty, and based on condition (ii), there exists a path from each agent $v\in\mathcal{S}(\theta^{\star},\theta)$ to every agent  $i\in \mathcal{V}\setminus\{v\}$;  Eq.  \eqref{eqn:asymprate} then follows from Lemma \ref{lemma:main}. By definition of a source set, $K_v(\theta^{\star},\theta)>0, \forall v\in\mathcal{S}(\theta^{\star},\theta)$; Eq.  \eqref{eqn:asymprate} then implies $\lim_{t\to\infty}\mu_{i,t}(\theta)=0$ a.s., $\forall i\in\mathcal{V}$.
\end{proof}

\begin{proof} (\textbf{Proposition \ref{prop:sparsity}})
Let the set $\bar{\Omega}$ have the same meaning as in Lemma \ref{lemma:main}. Fix any $\omega\in\bar{\Omega}$, and note that since the conditions of Theorem \ref{thm:main} are met, $\mu_{i,t}(\theta^*)\rightarrow 1$ on $\omega,\forall i\in\mathcal{V}$. We prove the first claim of the proposition via contradiction. Accordingly, suppose the claim does not hold. Since there are only finitely many agents, this implies the existence of some $i\in\mathcal{V}$ and some $j\in\mathcal{N}_i$,  such that $i$ broadcasts its belief on $\theta^*$ to $j$ infinitely often, i.e., there exists a sub-sequence $\{t_{p_k}\}$ of $\{t_k\}$ at which the event-condition \eqref{eq:event} gets satisfied for $\theta^*$. From \eqref{eq:event}, $\mu_{i,t_{p_k}}(\theta^*) < \gamma^{k}\mu_{i,t_{p_0}}(\theta^*),\forall k\in\mathbb{N}_{+}$, where $\gamma\triangleq \gamma(t_{p_0})$. This implies $\lim_{k\to\infty} \mu_{i,t_{p_k}}(\theta^*)=0$, contradicting the fact that on $\omega$, $\lim_{t\to\infty} \mu_{i,t}(\theta^*)=1$. 

For establishing the second claim, fix $\omega\in\bar{\Omega}$,  $\theta\neq\theta^*$, and $i\notin\mathcal{S}(\theta^*,\theta)$. Since $i\notin\mathcal{S}(\theta^*,\theta)$, there exists $\tilde{t}_1 < \infty$ and $\bar{\eta}>0$, such that $\pi_{i,t}(\theta) \geq \bar{\eta}, \forall t \geq \tilde{t}_1.$ This follows from the fact that since $\theta$ is observationally equivalent to $\theta^*$ for agent $i$, the claim regarding $\pi_{i,t}(\theta^*)$ in Eq. \eqref{eqn:lowerbound} applies identically to $\pi_{i,t}(\theta)$. Note also that since the conditions of Theorem \ref{thm:main} are met, $\mu_{i,t}(\theta) \rightarrow 0$ on $\omega$. From \eqref{eq:mubardefn}, $\bar{\mu}_{i,t}(\theta) \rightarrow 0$ as well. Thus, there must exist some $\tilde{t}_2 < \infty$ such that $\min\{\bar{\mu}_{i,t}(\theta),\pi_{i,t+1}(\theta)\}=\bar{\mu}_{i,t}(\theta), \forall t \geq \tilde{t}_2$. Let $\tilde{t}=\max\{\tilde{t}_1,\tilde{t}_2\}$. Consider any $t_k\in\mathbb{I}, t_k > \tilde{t}$. We claim:
\begin{equation}
    \mu_{i,t}(\theta) \geq \bar{\mu}_{i,t_k}(\theta), \forall t\in [t_{k}+1,t_{k+1}], \hspace{1mm} \textrm{and}
    \label{eq:uninf1}
\end{equation}
\begin{equation}
 \hspace{-12mm}   \bar{\mu}_{i,t}(\theta) \geq \bar{\mu}_{i,t_k}(\theta), \forall t\in [t_{k},t_{k+1}).
    \label{eq:uninf2}
\end{equation}
To see why the above inequalities hold, consider the update of $\mu_{i,t_k+1}(\theta)$ based on \eqref{eqn:update}. Since $t_k > \tilde{t}_2$, we have  $\min\{\bar{\mu}_{i,t_k}(\theta),\pi_{i,t_k+1}(\theta)\}=\bar{\mu}_{i,t_k}(\theta)$. Noting that the denominator of the fraction on the R.H.S. of \eqref{eqn:update} is at most $1$, we obtain: $\mu_{i,t_k+1}(\theta) \geq \bar{\mu}_{i,t_k}(\theta).$ If $t_{k}+1=t_{k+1}$, then the claim follows. Else, if $t_{k}+1 <t_{k+1}$, then since no communication occurs at $t_{k}+1$, we have from \eqref{eq:mubardefn} that $\bar{\mu}_{i,t_k+1}(\theta)=\min\{\bar{\mu}_{i,t_k}(\theta),\mu_{i,t_k+1}(\theta)\} \geq \bar{\mu}_{i,t_k}(\theta).$ We can keep repeating the above argument for each $t\in [t_k+1,t_{k+1}]$ to establish the claim. In words, inequalities \eqref{eq:uninf1} and \eqref{eq:uninf2} reveal that agent $i$ cannot lower its belief on the false hypothesis $\theta$ between two consecutive event-monitoring time-steps when it does not hear from any neighbor. We will make use of this fact repeatedly during the remainder of the proof. Let $t_{p_0}> \tilde{t}$ be the first time-step in $\mathbb{I}$ to the right of $\tilde{t}$. Now consider the following sequence, where $k\in\mathbb{N}$: 
\begin{equation}
    t_{p_{k+1}}=\inf\{t\in\mathbb{I}: t > t_{p_k}, \bar{\mu}_{i,t}(\theta) < \bar{\mu}_{i,t-1}(\theta) \}.
    \label{eq:Sequence}
\end{equation}
The above sequence represents those event-monitoring time-steps at which $\bar{\mu}_{i,t}(\theta)$ decreases. We first argue that $\{t_{p_k}\}$ is well-defined, i.e., each term in the sequence is finite. If not, then based on \eqref{eq:uninf2}, this would mean that $\bar{\mu}_{i,t}(\theta)$ remains bounded away from $0$, contradicting the fact that $\bar{\mu}_{i,t}(\theta) \rightarrow 0$ on $\omega$. Next, for each $k\in\mathbb{N}_{+}$, let $j_{p_k}\in \argmin_{j\in\mathcal{N}_{i,t_{p_k}}(\theta)\cup\{i\}} \mu_{j,t_{p_k}}(\theta).$ We claim that $i\neq j_{p_k}$. To see why this is true, suppose $i=j_{p_k}$. Then, based on the definition of $t_{p_k}$, we would have $\bar{\mu}_{i,t_{p_k}}(\theta)=\mu_{i,t_{p_k}}(\theta) < \bar{\mu}_{i,t_{p_k}-1}(\theta)$. However, as $t_{p_k} > \tilde{t}_2$, we have from \eqref{eqn:update} that $\mu_{i,t_{p_k}}(\theta) \geq \bar{\mu}_{i,t_{p_k}-1}(\theta)$, leading to the desired contradiction. In the final step of the proof, we claim that $i$ does not broadcast its belief on $\theta$ to $j_{p_k}$ over $[t_{p_k}+1,t_{p_{k+1}}]$. 

To establish this claim, we start by noting that based on the definitions of $j_{p_k}$ and $t_{p_k}$, $\bar{\mu}_{i,t_{p_k}}(\theta)=\mu_{j_{p_k},t_{p_k}}(\theta)$.  Let us first consider the case when there are no intermediate event-monitoring time-steps in $(t_{p_k},t_{p_{k+1}})$, i.e., $t_{p_k}$ and $t_{p_{k+1}}$ are consecutive terms in $\mathbb{I}$. 
Then, at $t_{p_{k+1}}$, $\hat{\mu}_{j_{p_k}i, t_{p_{k+1}}}(\theta)=\mu_{j_{p_k},t_{p_k}}(\theta)$, since no communication occurs over $(t_{p_k},t_{p_{k+1}})$. Moreover, using \eqref{eq:uninf1}, $\mu_{i,t_{p_{k+1}}}(\theta) \geq \bar{\mu}_{i,t_{p_k}}(\theta)=\mu_{j_{p_k},t_{p_k}}(\theta)$. Thus, the event condition \eqref{eq:event} gets violated at $t_{p_{k+1}}$, and $i$ does not broadcast its belief on $\theta$ to $j_{p_k}$. Next, consider the scenario when there is exactly one event-monitoring time-step - say $\bar{t}\in\mathbb{I}$ -  in the interval $(t_{p_k},t_{p_{k+1}})$.  Since $t_{p_k}$ and $\bar{t}$ are now consecutive terms in $\mathbb{I}$, the fact that $\mathds{1}_{ij_{p_k},\bar{t}}(\theta) \neq 1$ follows from exactly the same reasoning as earlier. We argue that $\mathds{1}_{j_{p_k}i,\bar{t}}(\theta) \neq 1$ as well. To see this, suppose that $j_{p_k}$ does in fact broadcast $\mu_{j_{p_k},\bar{t}}(\theta)$ to $i$ at $\bar{t}$. For this to happen, the event condition \eqref{eq:event} entails: $\mu_{j_{p_k},\bar{t}}(\theta) < \gamma(\bar{t}) \mu_{j_{p_k},t_{p_k}}(\theta)=\gamma(\bar{t}) \bar{\mu}_{i,t_{p_k}}(\theta) \leq \bar{\mu}_{i,t_{p_k}}(\theta)$. Since $\bar{\mu}_{i,\bar{t}-1}(\theta) \geq \bar{\mu}_{i,t_{p_k}}(\theta)$ from \eqref{eq:uninf2}, $\mathds{1}_{j_{p_k}i,\bar{t}}(\theta) = 1$ would then imply that $\bar{\mu}_{i,\bar{t}}(\theta) < \bar{\mu}_{i,\bar{t}-1}(\theta)$, violating the fact that $\bar{t} < t_{p_{k+1}}$. The above reasoning suggests that $\hat{\mu}_{j_{p_k}i, t}(\theta)=\mu_{j_{p_k},t_{p_k}}(\theta), \forall t\in (t_{p_k},t_{{p}_{k+1}}]$. Moreover, since $\bar{\mu}_{i,t}(\theta)$ does not decrease at $\bar{t}$ (as $\bar{t} < t_{p_{k+1}}$), we have from \eqref{eq:uninf1} that  $\mu_{i,t}(\theta) \geq \bar{\mu}_{i,t_{p_k}}(\theta)=\mu_{j_{p_k},t_{p_k}}(\theta), \forall t\in (t_{p_k},t_{p_{k+1}}]$. It follows from the preceding discussion that $\eqref{eq:event}$ gets violated at $t_{p_{k+1}}$, and hence $\mathds{1}_{ij_{p_k},t_{p_{k+1}}}(\theta) \neq 1$. The above arguments readily carry over to the case when there are an arbitrary number of event-monitoring time-steps in the interval $(t_{p_k},t_{p_{k+1}})$. Thus, we omit such details. 

We conclude that over each interval of the form $(t_{p_k},t_{p_{k+1}}], k \in \mathbb{N}_{+}$, there exists a neighbor $j_{p_k}\in\mathcal{N}_i$ to which agent $i$ does not broadcast its belief on $\theta$. We can obtain one such $t_{p_1}$ for each $i\notin\mathcal{S}(\theta^*,\theta)$, and take the maximum of such time-steps to obtain $T_2(\omega)$.
\end{proof}

\begin{proof} (\textbf{Corollary  \ref{prop:tree}}) 
Let us fix $\theta \neq \theta^*$, and partition the set of agents $\mathcal{V}\setminus\{v_{\theta}\}$ based on their distances from $v_{\theta}$. Accordingly, we use  $\mathcal{L}_q(\theta)$ to represent level-$q$ agents that are at distance $q$ from $v_{\theta}$, where $q\in\mathbb{N}_{+}.$ Let the agent(s) that are farthest from $v_{\theta}$ be at level $\bar{q}$. Now consider any agent $i\in\mathcal{L}_{\bar{q}}(\theta)$. Based on the conditions of the proposition, note that $i\notin\mathcal{S}(\theta^*,\theta)$, and the only neighbor of $i$ is its parent in the tree rooted at $v_{\theta}$, denoted by  $p_i(\theta)$. Thus, claim (ii) of Proposition \ref{prop:sparsity} applies to agent $i$, implying that agent $i$ stops broadcasting its belief on $\theta$ to $p_i(\theta)$ eventually almost surely. Next,  consider an agent $j\in\mathcal{L}_{\bar{q}-1}(\theta)$. We have already argued that after a finite number of time-steps, $j$ will stop hearing broadcasts about $\theta$ from its children in level $\bar{q}$. Thus, for large enough $k$, $\mathcal{N}_{j,t_k}(\theta)$ can only comprise of $p_j(\theta)$, namely the parent of agent $j$ in level $\bar{q}-2$. In particular, given that $j\notin\mathcal{S}(\theta^*,\theta)$, the decrease in $\bar{\mu}_{j,t}(\theta)$ at time-steps defined by \eqref{eq:Sequence} can only be caused by $p_j(\theta)$. It then readily follows from the proof of Proposition \ref{prop:sparsity} that $j$ will stop broadcasting $\mu_{j,t}(\theta)$ to $p_j(\theta)$ eventually almost surely. We can essentially keep repeating the above argument until we reach level 1.  
\end{proof}
\begin{proof} (\textbf{Theorem}  \ref{thm:asymp}) The proof of this result is similar in spirit to that of Theorem \ref{thm:main}. Hence, we only sketch the essential details. We begin by noting that the claims in Lemma \ref{lemma:bound} hold under the conditions of the theorem - this can be easily verified. Let $\bar{\Omega}$ have the same meaning as in Lemma \ref{lemma:main}. Fix $\omega\in\bar{\Omega}$ and an arbitrarily small $\epsilon > 0$. Since $\mathbb{P}^{\theta^*}(\bar{\Omega})=1$, to prove the result, it suffices to argue that for each false hypothesis $\theta\neq\theta^*$, $\exists T(\omega,\theta,\epsilon)$ such that on $\omega$,  $\mu_{i,t}(\theta) < \epsilon, \forall t\geq T(\omega,\theta,\epsilon), \forall i\in \mathcal{V}$. Recall that based on Lemma \ref{lemma:bound},  there exists a time-step $t'(\omega) < \infty$, and a constant $\eta(\omega)>0$, such that on $\omega$, $\pi_{i,t}(\theta^{\star}) \geq \eta(\omega), \bar{\mu}_{i,t}(\theta^{\star}) \geq \eta(\omega), \forall t\geq t'(\omega), \forall i\in\mathcal{V}$. Set $\bar{\epsilon}(\omega)=\min\{\epsilon,\gamma\eta(\omega)\}$. Also, from Lemma \ref{lemma:bound}, we know that there exists $\bar{t}$ such that $\mu_{i,t}(\theta) < \bar{\epsilon}^{|\mathcal{V}|}, \forall t\geq \bar{t}, \forall i\in\mathcal{S}(\theta^*,\theta)$.\footnote{As before, we have suppressed dependence of various quantities on $\omega,\theta,$ and $\epsilon$, since they can be inferred from context.} Let $\tilde{t}_0=\max\{t',\bar{t}\}.$ Since the union graph over $[\tilde{t}_0,\infty)$ is rooted at $\mathcal{S}(\theta^*,\theta)$, there exists a set $\mathcal{F}_1(\theta)\in\mathcal{V}\setminus\mathcal{S}(\theta^*,\theta)$ of agents such that each agent in $\mathcal{F}_1(\theta)$ has at least one neighbor in $\mathcal{S}(\theta^*,\theta)$ in the union graph. Accordingly, consider any $j\in\mathcal{F}_1(\theta)$, and suppose $j\in\mathcal{N}_i(\tau)$, for some $i\in\mathcal{S}(\theta^*,\theta)$, and some $\tau \geq \tilde{t}_0$. The cases $\mathds{1}_{ij,\tau}(\theta)=1$ and $\mathds{1}_{ij,\tau}(\theta)\neq 1$ can be analyzed exactly as in the proof of Lemma \ref{lemma:main} to yield: 
\begin{equation}
\mu_{j,t}(\theta) < \frac{\bar{\epsilon}^{|\mathcal{V}|}}{\eta \gamma} \leq \bar{\epsilon}^{(|\mathcal{V}|-1)}, \forall t > \tau,
\end{equation}
where the last inequality follows by noting that $\bar{\epsilon} \leq \eta\gamma$.
Let $\tilde{t}_1 > \tilde{t}_{0}$ be the first time-step by which every agent in $\mathcal{F}_1(\theta)$ has had at least one neighbor in $\mathcal{S}(\theta^*,\theta)$. Then, based on the above reasoning, $\mu_{j,t}(\theta) < \bar{\epsilon}^{(|\mathcal{V}|-1)}, \forall t > \tilde{t}_1, \forall j\in\mathcal{F}_1(\theta)$. If $\mathcal{V}\setminus\{\mathcal{S}(\theta^*,\theta)\cup\mathcal{F}_1(\theta)\}=\emptyset$, then we are done. Else, given the fact that the union graph over $[\tilde{t}_1,\infty)$ is rooted at $\mathcal{S}(\theta^*,\theta)$, there must exist a non-empty set $\mathcal{F}_2(\theta)$ such that each agent in $\mathcal{F}_2(\theta)$ has at least one neighbor from the set $\mathcal{S}(\theta^*,\theta)\cup\mathcal{F}_1(\theta)$ in the union graph. Reasoning as before, one can conclude that there exists a time-step $\tilde{t}_2 > \tilde{t}_1$ such that $\mu_{j,t}(\theta) < \bar{\epsilon}^{(|\mathcal{V}|-2)}, \forall t > \tilde{t}_2, \forall j\in\mathcal{F}_2(\theta)$. To complete the proof, we can keep repeating the above construction until we exhaust the vertex set $\mathcal{V}$.
\end{proof}
\section{Proof of Theorem \ref{thm:Quantmin}}
\label{app:ProofThmQmin}
We begin with the following lemma.
\begin{lemma}
Suppose the conditions of Theorem \ref{thm:Quantmin} are satisfied. Then, assertions (i)-(iii) in Lemma \ref{lemma:bound}
hold when each agent employs Algorithm  \ref{algo:Qmin}.
\label{lemma:qminbound}
\end{lemma}

\begin{proof}
The proof of this lemma mirrors that of Lemma \ref{lemma:bound}. The key point is that for any agent $i\in\mathcal{V}$, $q_{i,t}(\theta^*)\neq 0$ almost surely, where $t\in\mathbb{N}$. To see this, observe from $\eqref{eqn:encoder}$ that whenever an agent $i$ broadcasts about $\theta^*$, we have $q_{i,t}(\theta^*) \geq \mu_{i,t}(\theta^*)$. Hence, at such a time-step $t$, $q_{i,t}(\theta^*)=0 \implies \mu_{i,t}(\theta^*)=0$. Using the same arguments as in Lemma \ref{lemma:bound}, one can argue that this is almost surely impossible. 
\end{proof}

We are now ready to  prove Theorem \ref{thm:Quantmin}.

\begin{proof} (\textbf{Theorem \ref{thm:Quantmin}}) In view of Lemma \ref{lemma:qminbound}, we know that there exists a set $\bar{\Omega}\subseteq \Omega$ of $\mathbb{P}^{\theta^*}$-measure 1 for which assertions (ii) and (iii) of Lemma \ref{lemma:bound} hold. Consider any false hypothesis $\theta\neq\theta^*$, fix a sample path  $\omega\in\bar{\Omega}$, and an agent $v\in\mathcal{S}(\theta^{\star},\theta)$. Following the same reasoning as in the proof of Lemma \ref{lemma:main}, there exists a time-step $\bar{t}$, such that for all $t\geq \bar{t}$, the following are true on $\omega$: (i) $\pi_{i,t}(\theta^{\star}) \geq \eta(\omega), \bar{\mu}_{i,t}(\theta^{\star}) \geq \eta(\omega), \forall i\in\mathcal{V}$; and (ii) for a fixed $\epsilon>0$,  $\mu_{v,t}(\theta) <  e^{-(K_v(\theta^{\star},\theta)-\epsilon)t}$. We will complete the proof in two steps. In Step 1, we will establish that the quantization range $R_{v,t}(\theta)=[0,q_{v,t}(\theta)]$ contracts exponentially fast. In Step 2, we will analyze the implications of the above phenomenon on the beliefs of the remaining agents on $\theta$. In what follows, we elaborate on these steps. 

\textbf{Step 1.} Consider any time-step $t+1 > \bar{t}$. At this time-step, there are two possibilities. The first possibility is that  $\mu_{v,t+1}(\theta)\in[0,q_{v,t}(\theta))$, in which case we have from \eqref{eqn:encoder} that:
\begin{equation}
\resizebox{0.7\hsize}{!}{$
\begin{aligned}
q_{v,t+1}(\theta)&=\frac{q_{v,t}(\theta)}{2^{B(\theta)}}\ceil{{\mu_{v,t+1}(\theta)2^{B(\theta)}}/{q_{v,t}(\theta)}}\\
& < \frac{q_{v,t}(\theta)}{2^{B(\theta)}}\left(1+\frac{\mu_{v,t+1}(\theta)2^{B(\theta)}}{q_{v,t}(\theta)}\right)\\
& < \frac{1}{2^{B(\theta)}}q_{v,t}(\theta)+\mu_{v,t+1}(\theta).
\end{aligned}
$}
\label{eqn:quantbound}
\end{equation}
The second possibility is that  $\mu_{v,t+1}(\theta) \geq q_{v,t}(\theta)$ and, based on our encoding strategy, node $v$ sets $q_{v,t+1}(\theta)=q_{v,t}(\theta)$. Clearly, the bound on $q_{v,t+1}(\theta)$ in \eqref{eqn:quantbound} applies to both the cases we discussed above. To proceed, let $a=1/2^{B(\theta)}$, $\tilde{K}=K_v(\theta^*,\theta)-\epsilon$, and $\rho=\max\{a,e^{-\tilde{K}}\}$. Rolling out the inequality \eqref{eqn:quantbound} over $\tau\geq1$ time-steps starting from $\bar{t}$ yields:
\begin{equation}
\resizebox{0.85\hsize}{!}{$
    \begin{aligned}
q_{v,\bar{t}+\tau}(\theta)& < a^{\tau}\left(q_{v,\bar{t}}(\theta)+\sum\limits_{l=0}^{\tau-1}\frac{\mu_{v,\bar{t}+l+1}}{a^{l+1}}\right)\\
&\overset{(a)}{<} a^{\tau}\left(q_{v,\bar{t}}(\theta)+\frac{e^{-\tilde{K}(\bar{t}+1)}}{a}\sum\limits_{l=0}^{\tau-1}\frac{1}{{(ae^{\tilde{K}})}^l}\right)\\
&\overset{(b)}{<} a^{\tau}+\frac{e^{-\tilde{K}\tau}-a^{\tau}}{e^{-\tilde{K}}-a}\\
&\overset{(c)}{<}\left(1+\frac{1}{\vert e^{-\tilde{K}}-a \vert}\right){\rho}^{\tau}.
\end{aligned}
$}
\end{equation}
In the above inequalities, (a) follows by noting that $\mu_{v,\bar{t}+l+1}$ decays exponentially $\forall l\geq0$ based on the definition of $\bar{t}$. For (b), we simplify the preceding inequality using the facts that $q_{v,\bar{t}}(\theta) \leq 1$, and $e^{-\tilde{K}(\bar{t}+1)} \leq 1$ as $\tilde{K}>0$; the latter is true since $v\in\mathcal{S}(\theta^*,\theta)$. Finally, (c) follows from straightforward algebra. We thus obtain:
\begin{equation}
    q_{v,t}(\theta) <  \frac{1}{\rho^{\bar{t}}}\left(1+\frac{1}{\vert e^{-\tilde{K}}-a \vert}\right){\rho}^{t}, \forall t\geq \bar{t}+1.
\label{eqn:finalqbound}
\end{equation}
Since $B(\theta) \geq 1$, we have $a <1$. Moreover, as $\tilde{K} >0$, it follows that $\rho <1$. In view of \eqref{eqn:finalqbound}, we thus observe that $q_{v,t}(\theta)$ eventually decays to $0$ exponentially fast at the rate $\rho$. 

\textbf{Step 2.} Consider any neighbor $j$ of agent $v$. Let us now make two simple observations, each of which follow easily from the rules of Algorithm \ref{algo:Qmin}. First, given that $\mu_{v,1}(\theta) < 1=q_{v,0}(\theta)$, the condition in line 4 of Algorithm \ref{algo:Qmin} will pass at $t=1$, and hence agent $v$ will broadcast $q_{v,1}(\theta)$ to agent $j$ at time-step $t=1$. Second, at each subsequent  time-step $t\geq 1$, the value of $q_{v,t}(\theta)$ held by agent $v$ is consistent with that held by agent $j$, irrespective of whether $v$ broadcasts to $j$ at time $t$ about $\theta$, or not. We thus have that $\forall t\geq \bar{t}+2$:
\begin{equation}
\resizebox{0.6\hsize}{!}{$
    \begin{aligned}
\mu_{j,t}(\theta) &\overset{(a)}{\leq} \frac{\bar{\mu}_{j,t-1}(\theta)}{\eta}\\
& \overset{(b)}\leq \frac{q_{v,t-1}(\theta)}{\eta}\\
& \overset{(c)}<  \frac{1}{\eta\rho^{\bar{t}+1}}\left(1+\frac{1}{\vert e^{-\tilde{K}}-a \vert}\right){\rho}^{t},
    \end{aligned}
$}
\label{eqn:qfinal}
\end{equation}
where (a) follows from \eqref{eqn:update} and the fact that all beliefs on $\theta^*$ are bounded below by $\eta$ for $t\geq \bar{t}$; (b) follows from \eqref{eq:qmubar}; and (c) follows from \eqref{eqn:finalqbound}. Taking the natural log on both sides of \eqref{eqn:qfinal}, dividing throughout by $t$, and then taking the limit inferior on both sides of the resulting inequality yields:
\begin{equation}
    \liminf\limits_{t\to\infty}-\frac{\log\mu_{j,t}(\theta)}{t} \geq \log{\frac{1}{\rho}}.
\label{eqn:qrate}
\end{equation}
Now let us consider two cases. First, suppose $B(\theta)\log2 \geq K_v(\theta^*,\theta)$. Then, $\log1/\rho = \tilde{K}=K_v(\theta^*,\theta)-\epsilon$, where $\epsilon$ can be made arbitrarily small. Hence, in this case, the L.H.S. of \eqref{eqn:qrate} is at least $K_v(\theta^*,\theta)$. Next, suppose $B(\theta)\log2 < K_v(\theta^*,\theta)$. Then, there must exist $\epsilon > 0$ such that $B(\theta)\log2 < K_v(\theta^*,\theta)-\epsilon$. With such a choice of $\epsilon$, we can set $\tilde{K}=K_v(\theta^*,\theta)-\epsilon$ and conduct the above analysis to arrive at $\log1/ \rho = B(\theta)\log2$. We conclude:  
\begin{equation}
    \liminf\limits_{t\to\infty}-\frac{\log\mu_{j,t}(\theta)}{t} \geq \min\{B(\theta)\log2, K_v(\theta^*,\theta)\}.
\label{eqn:qfinalrate}
\end{equation}
Consider any neighbor $l$ of agent $j$, i.e., a two-hop neighbor of agent $v$. We can analyze the decay of $q_{j,t}(\theta)$ and $\mu_{l,t}(\theta)$ exactly as we did for $q_{v,t}(\theta)$ and $\mu_{j,t}(\theta)$ to conclude that $\mu_{l,t}(\theta)$ also decays exponentially at a rate that is lower bounded by $H_v(\theta^*,\theta)=\min\{B(\theta)\log2, K_v(\theta^*,\theta)\}$; this is not too hard to verify and hence we omit details. Repeating this argument reveals that every agent reachable from $v$ can reject $\theta$ at a rate that is at least $H_v(\theta^*,\theta)$. Since $\mathcal{G}$ is connected, the above conclusion applies to every agent.

An analysis identical to the one above can be carried out for each  $v\in\mathcal{S}(\theta^*,\theta)$. The proof can then be completed following the same arguments as in Theorem \ref{thm:main}.
\end{proof}

\section{Proof of Theorem \ref{thm:event_quant}}
\label{app:event_quant}

{
\begin{proof} (\textbf{Theorem \ref{thm:event_quant}}) The proof of this result is a simple variation on that of Theorem \ref{thm:Quantmin}. Hence, we will be somewhat terse in our arguments. First, it is easy to verify that Lemma \ref{lemma:qminbound} holds for the \texttt{QET} Min-Rule as well. Accordingly, let $\bar{\Omega}$ have the same meaning as in Theorem \ref{thm:Quantmin}, and fix a sample path  $\omega\in\bar{\Omega}$. Next, consider any  $\theta\neq\theta^*$, and $v\in\mathcal{S}(\theta^{\star},\theta)$. Fix $\epsilon > 0$, and recall from Theorem \ref{thm:Quantmin} that there exists a time-step $\bar{t}_1$, such that for all $t\geq \bar{t}_1$, the following are true on $\omega$: (i) $\pi_{i,t}(\theta^{\star}) \geq \eta(\omega), \bar{\mu}_{i,t}(\theta^{\star}) \geq \eta(\omega), \forall i\in\mathcal{V}$; and (ii) for a fixed $\epsilon>0$,  $\mu_{v,t}(\theta) <  e^{-(K_v(\theta^{\star},\theta)-\epsilon)t}$. From the condition on the threshold function $\gamma(\cdot)$ in Eq. \eqref{eqn:functions}, we also know that there exists $\bar{t}_2$ such that
\begin{equation}
\frac{1}{\gamma(t)} \leq e^{\epsilon t}, \forall t \geq \bar{t}_2.  
\label{eqn:gamma}
\end{equation}
Let $\bar{t}=\max\{\bar{t}_1,\bar{t}_2\}$, and let $t_p \in \mathbb{I}$ be the first event-monitoring time-step satisfying $t_p \geq \bar{t}$. Now consider any $t_k \in \mathbb{I}$ such that $t_k > t_p$. At $t_k$, if the event condition \eqref{eqn:event_quant} holds, then $\mu_{v,t_k}(\theta)$ is quantized to $q_{v,t_k}(\theta)$ based on the encoder in \eqref{eqn:encoder}. This yields:
\begin{equation}
    q_{v,t_k}(\theta) \leq \frac{q_{v,t_{k-1}}(\theta)}{2^{B(\theta)}} + \mu_{v,t_k}(\theta).
\label{eqn:quantrangebnd1}
\end{equation}
In the above step, we used $q_{v,t_{k-1}}(\theta)=q_{v,t_k-1}(\theta)$ by noting that $q_{v,t}(\theta)$ does not change over the interval $[t_{k-1}, t_k -1]$. If the event condition \eqref{eqn:event_quant} fails at time $t_k$, we would then have
\begin{equation}
    q_{v,t_k}(\theta) = q_{v,t_{k-1}(\theta)} \leq \frac{\mu_{v,t_k}(\theta)}{\gamma(t_k)}.
\label{eqn:quantrangebnd2}
\end{equation}
Combining the bounds in \eqref{eqn:quantrangebnd1} and \eqref{eqn:quantrangebnd2}, we obtain
\begin{equation}
    q_{v,t_k}(\theta) \leq \frac{q_{v,t_{k-1}}(\theta)}{2^{B(\theta)}} + \frac{\mu_{v,t_k}(\theta)}{\gamma(t_k)}.
\label{eqn:quantrangebnd}
\end{equation}
Our immediate goal is to analyze the above periodic  recursion and show that the quantizer range $R_{v,t}(\theta) = [0, q_{v,t}(\theta)]$ shrinks exponentially fast. To proceed, let $a=1/2^{B(\theta)}$, $\bar{a}=a^{1/\tau}$, $\tilde{K}=K_v(\theta^*,\theta)-2\epsilon$, and $\rho=\max \{e^{-\tilde K}, \bar{a} \}.$ Next, observe that for any positive integer $h$,
\begin{equation}
    \begin{aligned}
    q_{v,t_{p+h}(\theta)} & \leq a^h \left(q_{v,t_p}(\theta) + \sum_{\ell=1}^h a^{-\ell} \frac{ \mu_{v,t_{p+\ell}}(\theta)}{\gamma_{t_{p+\ell}}} \right) \\
  &\overset{(a)}\leq a^h \left(q_{v,t_p}(\theta) + \sum_{\ell=1}^h a^{-\ell} e^{-\tilde{K} t_{p+\ell}} \right) 
    \\
 &\overset{(b)} = a^h \left(q_{v,t_p}(\theta) + e^{-\tilde{K} t_p}\sum_{\ell=1}^h {\left(\frac{1}{a e^{\tilde{K} \tau}}\right)}^l 
 \right) \\
 &\leq \bar{a}^{\tau h} + \frac{\vert e^{-\tilde K \tau h} - \bar{a}^{\tau h} \vert}{\vert e^{-\tilde K \tau} - \bar{a}^{\tau} \vert} \\
 & \leq \frac{1}{\rho^{t_p}} \left(1 + \frac{1}{\vert e^{-\tilde K \tau} - \bar{a}^{\tau} \vert}\right) \rho^{t_{p+h}}. 
 \end{aligned}
\end{equation}
For (a), we used the fact that $t_{p+h} \geq t_p \geq \bar{t}$, and \eqref{eqn:gamma}. For (b), we used $t_{p+\ell}=t_p + \tau \ell$.\footnote{Recall that the event-monitoring sequence is periodic with period $\tau$.} We conclude that for any $t_k \in \mathbb{I}$ such that $t_k \geq t_p$, 
\begin{equation}
    q_{v,t_k}(\theta) \leq \frac{1}{\rho^{t_p}} \left(1 + \frac{1}{\vert e^{-\tilde K \tau} - \bar{a}^{\tau} \vert}\right) \rho^{t_{k}}. 
\label{eqn:QETbound1}
\end{equation}
Consider any $t\geq t_{p+1}$, and let $t_{f(t)}=\max\{t_k \in \mathbb{I} : t_k \leq t\}$. Noting that the sequence $\{q_{v,t}(\theta)\}$ is non-increasing (based on the rules of the \texttt{QET Min-Rule}), and using \eqref{eqn:QETbound1}, we obtain
\begin{equation}
\begin{aligned}
    q_{v,t}(\theta) &\leq q_{v,t_{f(t)}}(\theta) \leq \frac{1}{\rho^{t_p}} \left(1 + \frac{1}{\vert e^{-\tilde K \tau} - \bar{a}^{\tau} \vert}\right) \rho^{t+t_{f(t)}-t} \\ 
      & \leq \underbrace{\frac{1}{\rho^{t_{p+1}}} \left(1 + \frac{1}{\vert e^{-\tilde K \tau} - \bar{a}^{\tau} \vert}\right)}_{C} \rho^t.
\end{aligned}
\label{eqn:QETbound2}
\end{equation}
To arrive at the last inequality, we used: (i) $\rho < 1$; (ii) $t-t_{f(t)} \leq \tau$; and (iii) $t_p +\tau=t_{p+1}$. Now consider any $j\in\mathcal{N}_v$, $t > t_{p+1}$, and let $t_{\ell(t)}=\max\{t_k \in \mathbb{I} : t_k < t\}$. At $t_{\ell(t)}$, $v$ either broadcasts $q_{v,t_{\ell(t)}}(\theta)$ to $j$ or it does not, depending upon whether or not the event condition \eqref{eqn:event_quant} holds. Each of these cases can be analyzed just as in the proof of Theorem \ref{thm:main} to arrive at the following conclusion:
\begin{equation}
    \bar{u}_{j,t_{\ell(t)}}(\theta) \leq \frac{C}{\gamma(t)} \rho^{t_{\ell(t)}} \leq \frac{C}{\rho^{\tau}} \frac{\rho^t}{\gamma(t)}. 
\end{equation}

Using \eqref{eqn:update}, we then have
\begin{equation}
    \mu_{j,t}(\theta) \leq \frac{\bar{\mu}_{j,t-1}(\theta)}{\eta} \leq \frac{\bar{\mu}_{j,t_{\ell(t)}}(\theta)}{\eta} \leq \frac{C}{\eta \rho ^{\tau}} \frac{\rho^t}{\gamma(t)}. \end{equation}

Thus, for any $t > t_{p+1}$, we have 
\begin{equation}
    -\frac{\log\mu_{j,t}(\theta)}{t} \geq \log{\frac{1}{\rho}} -\frac{\log \bar{C}}{t} -\frac{\log(1/\gamma(t))}{t},
\end{equation}
where $\bar{C}=C/(\eta \rho^{\tau})$. Taking the limit inferior on both sides of the above inequality, and using \eqref{eqn:functions}, we obtain \eqref{eqn:qrate}. The rest of the proof is similar to that of Theorem \ref{thm:Quantmin}, and hence, we omit details in the interest of space.
\end{proof}
} \color{black}{}
\bibliographystyle{IEEEtran} 
\bibliography{refs}
\IEEEpeerreviewmaketitle
\end{document}